%% file: main.tex
\documentclass[11pt]{article}
\input{macros}

\author{
Grigory Yaroslavtsev\thanks{Indiana University, Bloomington, IN., USA \& The Alan Turing Institute, London, UK. \texttt{gyarosla@iu.edu}}
\and
Samson Zhou\thanks{Indiana University, Bloomington, IN., USA.
\texttt{samsonzhou@gmail.com}}
}
\date{\today}

\begin{document}
\title{Approximate $\ftwo$-Sketching of Valuation Functions}

\maketitle

\begin{abstract}
We study the problem of constructing a linear sketch of minimum dimension that allows approximation of a given real-valued function $f \colon \ftwo^n \rightarrow \mathbb R$ with small expected squared error. 
We develop a general theory of linear sketching for such functions through which we analyze their dimension for most commonly studied types of valuation functions: additive, budget-additive, coverage, $\alpha$-Lipschitz submodular and matroid rank functions. 
This gives a characterization of how many bits of information have to be stored about the input $x$ so that one can compute $f$ under additive updates to its coordinates.

Our results are tight in most cases and we also give extensions to the distributional version of the problem where the input $x \in \ftwo^n$ is generated uniformly at random. Using known connections with dynamic streaming algorithms, both upper and lower bounds on dimension obtained in our work extend to the space complexity of algorithms evaluating $f(x)$ under long sequences of additive updates to the input $x$ presented as a stream. 
Similar results hold for simultaneous communication in a distributed setting.
\end{abstract}

\input{intro}
\input{prelims}

\input{matroid}

\section*{Acknowledgments}
We would like to thank Swagato Sanyal for multiple discussions leading to this paper, including the proof of Theorem~\ref{thm:approx-f2-sketch-uniform} and Nikolai Karpov for his contributions to Section~\ref{sec:rank-2}. 
We would also like to thank Amit Chakrabarti, Qin Zhang and anonymous reviewers for their comments.

\bibliographystyle{alpha}
\bibliography{linsketch}
\appendix
\input{appendix}

\end{document}

%% file: macros.tex
\usepackage{multirow}
\usepackage{fullpage}
\usepackage[disable]{todonotes}
\usepackage{graphicx} 
\usepackage{array} 
\usepackage{paralist} 
\usepackage{amsthm, amsmath, amssymb, color, verbatim}
\usepackage{hyphenat}
\usepackage{epsfig}
\usepackage{subfigure}
\usepackage{multirow}
\usepackage{latexsym}
\usepackage{hyperref}
\usepackage{footnote}
\makesavenoteenv{table}
\makesavenoteenv{tabular}

\usepackage{algorithm2e}

\newtheorem{theorem}{Theorem}[section]
\newtheorem{lemma}[theorem]{Lemma}
\newtheorem{proposition}[theorem]{Proposition}
\newtheorem{corollary}[theorem]{Corollary}

\newtheorem{fact}[theorem]{Fact}

\newtheorem{question}[theorem]{Question}
\newtheorem{definition}[theorem]{Definition}

\newenvironment{proofof}[1]{\par\medskip\noindent{\bf Proof of #1: \ }}{\hfill$\Box$\par\medskip}

\renewcommand{\qed}{\nobreak \ifvmode \relax \else
      \ifdim\lastskip<1.5em \hskip-\lastskip
      \hskip1.5em plus0em minus0.5em \fi \nobreak
      \vrule height0.75em width0.5em depth0.25em\fi}

\newcommand{\eps}{\epsilon}

\def\E{\mathop{\mathbb{E}}\displaylimits}

\newcommand{\cA}{\mathcal A}
\newcommand{\cB}{\mathcal B}
\newcommand{\cD}{\mathcal D}
\newcommand{\cM}{\mathcal M}
\newcommand{\cS}{\mathcal S}
\newcommand{\Index}{\mathsf{INDEX}}

\newcommand{\R}[0]{{\ensuremath{\mathbb{R}}}}

\newcommand{\cI}[0]{\mathcal{I}}

\newcommand{\D}[0]{\textsf{D}}

\newcommand{\cE}[0]{\mathcal{E}}

\newcommand{\rl}[1]{R^{lin}_{#1}}
\newcommand{\rla}[1]{\bar{R}^{lin}_{#1}}
\newcommand{\rlad}[2]{\bar{\cD}^{lin, #1}_{#2}}

\newcommand{\distc}[2]{\bar{\mathcal D}^{\rightarrow, #2}_{#1}}

\newcommand{\distcu}[1]{\distc{#1}{U}}

\newcommand{\rc}[1]{R^{\rightarrow}_{#1}}
\newcommand{\rca}[1]{\bar{R}^{\rightarrow}_{#1}}

\newcommand{\F}{\mathbb F}

\newcommand{\oo}{\{+1,-1\}}
\newcommand{\ftwo}{\F_2}
\def\E{\mathop{\mathbb{E}}\displaylimits}

\newcommand{\sgn}[1]{\ensuremath{\mathsf{sgn}\left(#1\right)}}

\newcommand{\fplus}[1]{f^{+#1}}

\newcommand*{\rchi}{\mbox{\Large$\chi$}}

\newif\ifnotes\notestrue 
\ifnotes
\newcommand{\samson}[1]{\textcolor{purple}{{\bf (Samson:} {#1}{\bf ) }} \marginpar{\tiny\bf
             \begin{minipage}[t]{0.5in}
               \raggedright S:
            \end{minipage}}}            							
\else
\newcommand{\grigory}[1]{\textcolor{blue}{{\bf (Grigory:} {#1}{\bf ) }} \marginpar{\tiny\bf
             \begin{minipage}[t]{0.5in}
               \raggedright G:
            \end{minipage}}}            							
\else
\newcommand{\samson}[1]{}
\newcommand{\grigory}[1]{}
\fi

\newcommand{\superscript}[1]{\ensuremath{^{\mbox{\tiny{\textit{#1}}}}}\xspace}
\def \th {\superscript{th}}     

\hypersetup{
    colorlinks   = true,
    citecolor    = blue,
	linkcolor	 = red
}

%% file: intro.tex
\section{Introduction}
Linear sketching is a fundamental tool in efficient algorithm design that has enabled many of the recent breakthroughs in fast graph algorithms and computational linear algebra. 
It has a wide range of applications, including randomized numerical linear algebra (see survey~\cite{W14}), graph sparsification (see survey~\cite{M14}), frequency estimation \cite{AMS99}, dimensionality reduction \cite{JL84}, various forms of sampling, signal processing, and communication complexity. 
In fact, linear sketching has been shown to be the optimal algorithmic technique \cite{LNW14, AHLW16} for dynamic data streams, where elements can be both inserted and deleted. 
Linear sketching is also a frequently used tool in distributed computing --- summaries communicated between the processors in massively parallel computational models are often linear sketches. \todo{maybe mention the paper with UCSD people somewhere?}

In this paper we introduce a study of approximate linear sketching over $\ftwo$
(approximate $\ftwo$-sketching). 
This is a previously unstudied but natural generalization of the work of~\cite{KMSY18}, which studies exact $\ftwo$-sketching.
For a set $S \subseteq [n]$ let $\chi_{S}\colon \ftwo^n \to \ftwo$ be a parity function defined as $\chi_{S}(x) = \sum_{i \in S} x_i$.
Given a function $f \colon \ftwo^n \rightarrow \mathbb R$, we are looking for a distribution over $k$ subsets $\mathbf{S}_1, \dots, \mathbf{S}_k \subseteq [n]$ such that for any input $x$, it should be possible to compute $f(x)$ with expected squared error at most $\eps$ from the parities $\chi_{\mathbf{S}_1}(x), \chi_{\mathbf{S}_2}(x), \dots, \chi_{\mathbf{S}_k}(x)$ computed over these sets.
While looking only at linear functions over $\ftwo$ as candidate sketches for evaluating $f$ might seem restrictive, this view turns out to be optimal in a number of settings. In the light of recent results of~\cite{KMSY18,HLY18},  the complexity of $\ftwo$-sketching also characterizes the space complexity of streaming algorithms in the XOR-update model as well as communication complexity of one-way multiplayer broadcasting protocols for XOR-functions.

In matrix form, $\ftwo$-sketching corresponds to multiplication over $\mathbb F_2$ of the row vector $x \in \ftwo^n$ by a random $n \times k$ matrix  whose $i$-th column is the characteristic vector of $\chi_{\mathbf{S}_i}$:
\[
\bordermatrix {
	&&&&\cr
	& x_1     & x_2     & \ldots & x_n  \cr
}
\bordermatrix{
	&     &      &  &      \cr
	& \vdots & \vdots & \vdots & \vdots     \cr
	& \chi_{\mathbf{S}_1}     & \chi_{\mathbf{S}_2}     & \ldots & \chi_{\mathbf{S}_k}     \cr
	& \vdots & \vdots & \vdots & \vdots \cr
}
=\bordermatrix {
	&&&&\cr
	& \chi_{\mathbf{S}_1}(x)     & \chi_{\mathbf{S}_2}(x)     & \ldots & \chi_{\mathbf{S}_k}(x)  \cr
}
\]

The goal is to minimize $k$, ensuring that the sketch alone is sufficient for computing $f$ with expected squared error at most $\eps$ for any fixed input $x$. 
For a fixed distribution $\D$ of $x$, the definition of error is modified to include an expectation over $\D$ in the error guarantee. 
We give formal definitions below.

\begin{definition}[Exact $\ftwo$-sketching,~\cite{KMSY18}]\label{def:f2-sketch}
The \emph{exact randomized $\ftwo$-sketch complexity} with error $\delta$ of a function $f \colon \ftwo^n \to \mathbb R$ (denoted as $\rl{\delta}(f)$) is the smallest integer $k$ such that there exists a distribution $\chi_{\mathbf{S}_1},\chi_{\mathbf{S}_2},\ldots, \chi_{\mathbf{S}_k}$ over $k$ linear functions over $\ftwo^n$ and a post-processing function $g:\ftwo^k \rightarrow \mathbb R$ that satisfies:
	$$
	\forall x \in \ftwo^n \colon \Pr_{\mathbf{S}_1, \dots, \mathbf{S}_k}\left[
	g(\chi_{\mathbf{S}_1}(x),\chi_{\mathbf{S}_2}(x),\ldots, \chi_{\mathbf{S}_k}(x)) = f(x)\right] \ge 1-\delta.
	$$
\end{definition}

The number of parities $k$ in the definition above is referred to as the \textit{dimension} of the $\ftwo$-sketch.

\begin{definition}[Approximate $\ftwo$-sketching]\label{def:rand-f2-sketch}
The \emph{$\eps$-approximate randomized $\ftwo$-sketch complexity} of a function $f \colon \ftwo^n \to \mathbb R$ (denoted as $\rla{\eps}(f)$) is the smallest integer $k$ such that there exists a distribution $\chi_{\mathbf{S}_1},\chi_{\mathbf{S}_2},\ldots, \chi_{\mathbf{S}_k}$ over $k$ linear functions over $\ftwo^n$ and a post-processing function $g:\ftwo^k \rightarrow \mathbb R$ that satisfies:
	$$
	\forall x \in \ftwo^n \colon \E_{\mathbf{S}_1, \dots, \mathbf{S}_k}\left[
	(g(\chi_{\mathbf{S}_1}(x),\chi_{\mathbf{S}_2}(x),\ldots, \chi_{\mathbf{S}_k}(x)) - f(x))^2  \right] \le \eps 
	$$
\end{definition}

If $g$ is an unbiased estimator of $f$, then this corresponds to an upper bound on the variance of the estimator.
For example, functions with small spectral norm (e.g. coverage functions, Corollary~\ref{cor:coverage}) admit such approximate $\ftwo$-sketches. 
Moreover, observe that Definition~\ref{def:rand-f2-sketch} is not quite comparable with an epsilon-delta guarantee, which only promises that $|g(\chi_{\mathbf{S}_1}(x),\chi_{\mathbf{S}_2}(x),\ldots, \chi_{\mathbf{S}_k}(x)) - f(x)|\le\epsilon$ with probability $1-\delta$, but guarantees nothing for $\delta$ fraction of the inputs. 

In addition to this worst-case guarantee, we also consider the same problem for $x$ from a certain distribution. 
In this case, a weaker guarantee is required, i.e. the bound on expected squared error should hold only over some fixed known distribution $\D$.
An important case is $\D = U(\ftwo^n)$, the uniform distribution over all inputs.

\begin{definition}[Approximate distributional $\ftwo$-sketching]\label{def:rand-dist-f2-sketch}
	For a function $f \colon \ftwo^n \to \mathbb R$, we define its \emph{$\eps$-approximate randomized distributional $\ftwo$-sketch complexity} with respect to a distribution $\D$ over $\ftwo^n$
	(denoted as $\rlad{\D}{\eps}(f)$) as the smallest integer $k$ such that there exists a distribution $\chi_{\mathbf{S}_1},\chi_{\mathbf{S}_2},\ldots, \chi_{\mathbf{S}_k}$ over $k$ linear functions over $\ftwo$ and a post-processing function $g:\ftwo^k \rightarrow \ftwo$ that satisfies:
	$$
	\E_{x \sim \D} \E_{\mathbf{S}_1, \dots, \mathbf{S}_k}\left[
	(g(\chi_{\mathbf{S}_1}(x),\chi_{\mathbf{S}_2}(x),\ldots, \chi_{\mathbf{S}_k}(x)) - f(x))^2 \right] \le \eps. 
	$$
\end{definition}

\subsection{Applications to Streaming and Distributed Computing}
One of the key applications of our results is to the dynamic streaming model.
In this model, the input $x$ is generated via a sequence of additive updates to its coordinates, starting with $x = 0^n$.
If $x \in \mathbb R^n$, then updates are of the form $(i, \Delta_i)$ (turnstile model), where $i \in [n]$, and $\Delta_i \in \mathbb R$, which adds $\Delta_i$ to the $i$-th coordinate of $x$.
For $x \in \ftwo^n$, only the coordinate $i$ is specified and the corresponding bit is flipped, which is known as the XOR-update model~\cite{T16}\footnote{By slightly changing the function to $f'(x_1, \dots, x_n, y_1, \dots, y_n) = f(x_1 + y_1, x_2 + y_2, \dots, x_n + y_n)$, it is easy to see that there are functions for which knowledge of the sign of the update (i.e. whether it is +1 or -1) is not a stronger model than the XOR-update model. 
For some further motivation of the XOR-update model, consider dynamic graph streaming algorithms, i.e the setting when $x$ represents the adjacency matrix of a graph and updates correspond to adding and removing the edges. 
Almost all known dynamic graph streaming algorithms (except spectral graph sparsification of~\cite{KLMMS17}) are based on the $\ell_0$-sampling primitive~\cite{FIS08}. 
As shown recently, $\ell_0$-sampling can be implemented optimally using $\ftwo$-sketches~\cite{KNPWWY17} and hence almost all known dynamic graph streaming algorithms can handle XOR-updates, i.e. knowing whether an edge was inserted or deleted does not help.}.
Dynamic streaming algorithms aim to minimize space complexity of computing a given function $f$ for an input generated through a sequence of such updates while also ensuring fast update and function evaluation times.   

Note that linear sketching over the reals and $\mathbb F_2$-sketching can be used directly in the respective streaming update models. Most interestingly, these techniques turn out to achieve almost optimal space complexity. It is known that linear sketching over the reals gives (almost) optimal space complexity for processing dynamic data streams in the turnstile model for \textit{any} function $f$~\cite{LNW14,AHLW16}. 
However, the results of~\cite{LNW14,AHLW16} require adversarial streams of length triply exponential in $n$.
In the XOR-update model, space optimality of $\mathbb F_2$-sketching has been shown recently in~\cite{HLY18}. 
This optimality result holds even for adversarial streams of much shorter length $\Omega(n^2)$.
Hence, lower bounds on $\ftwo$-sketch complexity obtained in our work extend to space complexity of dynamic streaming algorithms for streams of quadratic length. 

A major open question in this area is the conjecture of~\cite{KMSY18} that the same holds even for streams of length only $2n$.
We thus complement our lower bounds on dimension of $\ftwo$-sketches with one-way two-player communication complexity lower bounds for the corresponding XOR functions $\fplus{}(x,y) = f(x + y)$. 
Such lower bounds translate to dynamic streaming lower bounds for streams of length $2n$.
Furthermore, whenever our communication lower bounds hold for the uniform distribution, the corresponding streaming lower bound applies to streaming algorithms under uniformly random input updates. 

Finally, our upper bounds can be used for distributed algorithms computing $f(x_1 + \dots + x_M)$ over a collection of distributed inputs $x_1, \dots, x_M \in \mathbb F_2^n$ as $\ftwo$-sketches can be used for distributed inputs. On the other hand, our communication lower bounds also apply to the simultaneous message passing (SMP) communication model, since it is strictly harder than one-way communication.

\subsection{Valuation Functions and Sketching}
Submodular valuation functions, originally introduced in the context of algorithmic game theory and optimization, have received a lot of interest recently in the context of learning theory~\cite{BH10,BCIW12,CKKL12,GHRU13,RY13,FKV13,FK14,FV15,FV16}\footnote{We remark that in this literature the term ``sketching'' is used to refer to the space complexity of representing the function $f$ itself under the assumption that it is unknown but belongs to a certain class. This question is orthogonal to our work as we assume $f$ is known and fixed while the input $x$ is changing. },  approximation~\cite{GHIM09,BDFKNR12} and property testing~\cite{CH12,SV14,BB17}.
As we show in this work, valuation functions also represent an interesting study case for linear sketching and streaming algorithms. 
While a variety of papers exists on streaming algorithms for optimizing various submodular objectives, e.g.~\cite{SG09,DIMV14,BMKK14,CGQ15,CW16,ER16,HIMV16,AKL16,BEM17}, to the best of our knowledge no prior work considers the problem of evaluating such functions under a changing input.

A systematic study of $\ftwo$-sketching has been initiated for Boolean functions in~\cite{KMSY18}. 
This paper can be seen as a next step, as we introduce approximation into the study of $\ftwo$-sketching. 
One of the consequences of our work is that the Fourier $\ell_1$-sampling technique, originally introduced by Bruck and Smolensky~\cite{BS92} (see also~\cite{G97,MO09}), turns out to be optimal in its dependence on both spectral norm and the error parameter. 
For Boolean functions, a corresponding result is not known as Boolean functions with small spectral norm and necessary properties are hard to construct.
Another technical consequence of our work is that the study of learning and sketching algorithms turn out to be related on a technical level despite pursuing different objectives (in learning the specific function is unknown, while in sketching it is). In particular, our hardness result for Lipschitz submodular functions uses a construction of a large family of matroids from~\cite{BH10} (even though in a very different parameter regime), who designed such a family to fool learning algorithms.

\subsection{Our Results}
A function $f \colon 2^{[n]} \to \mathbb R$ is $\alpha$-Lipschitz if for any $S \subseteq [n]$ and $i \in [n]$, it holds that $|f(S \cup \{i\}) - f(S)| \le \alpha$ for some constant $\alpha > 0$.
A function $f \colon 2^{[n]} \to \mathbb R$ is submodular if:
\begin{align*}
f(A \cup \{i\}) - f(A) \ge f(B \cup \{i\}) - f(B) &\quad\quad\quad \forall A \subseteq B \subseteq [n] \text{ and } i \notin B.
\end{align*}

We consider the following classes of valuation functions of the form $f \colon \ftwo^n \rightarrow \mathbb R$ (all of them submodular) sometimes treating them as $f \colon 2^{[n]} \rightarrow \mathbb R$ and vice versa. 
These classes mostly cover all of existing literature on submodular functions\footnote{We do not discuss some other subclasses of subadditive functions because they are either superclasses of classes for which we already have an $\Omega(n)$ lower bound (e.g. submodular, subadditive, etc.) or because such a lower bound follows trivially (e.g. for OXS/XOS since for XS-functions a lower bound of $\Omega(n)$ is easy to show, see Appendix~\ref{app:xs-functions}).}. See Table~\ref{table:results} for a summary of the results.

\begin{table}
	\centering

{\renewcommand{\arraystretch}{1.3}
	\begin{tabular}[!htb]{|c|c|c|c|c|}\hline
		Class & Error & Distribution & Complexity & Result \\\hline
		Additive/Budget additive  &\multirow{2}{*}{$\eps$} & \multirow{2}{*}{any} & \multirow{2}{*}{$\Theta\left(\frac{\|w\|_1^2}{\eps}\right)$}  & Theorem~\ref{thm:weighted-linear}, \ref{thm:lb-hockey-stick-uniform} \\
		$\min(b, \sum_{i = 1}^n w_i x_i)$ &   &  &  & Corollary~\ref{cor:additive},~\ref{cor:budget-additive}\\\hline
		$\min(c\sqrt{n}, \frac{2c}{\sqrt{n}} \sum_{i = 1}^n x_i)$&constant&uniform&$\Omega(n)$&Theorem~\ref{thm:lb-hockey-stick-uniform}\\\hline
		Coverage &$\eps$& any & $O\left(\frac{1}{\epsilon}\right)$ & Corollary~\ref{cor:coverage} \\\hline
		Matroid Rank 2  & exact & any & $\Theta(1)$ & Theorem~\ref{thm:rank-2-matroids}\\\hline
		Graphic Matroids Rank $r$  & exact & any & $O(r^2\log r)$ & Theorem~\ref{thm:graphic-matroids}\\\hline
		Matroid Rank $r$  & exact &any  & $\Omega(r)$ & Corollary~\ref{cor:lb:rank}\\\hline
		Matroid Rank $r$  & exact &uniform& $O((r \log r + c)^{r +1})$ & Corollary~\ref{cor:matroid-approximation-uniform-rank}\\\hline
		Matroid Rank & $1/\sqrt{n}$ &uniform&$\Theta(1)$ & Corollary~\ref{cor:matroid-approximation-uniform} \\\hline
		$\frac{c}{n}$-Lipschitz Submodular&constant& any & $\Theta(n)$ & Theorem~\ref{thm:lb-lipschitz-submodular}\\\hline
	\end{tabular}
\caption{Linear sketching complexity of classes of valuation functions}\label{table:results}
}
\end{table}

\begin{itemize}
\item \textbf{Additive (linear).} $f(x) = \sum_{i = 1}^n w_i x_i$, where $w_i \in \mathbb R$.

\textbf{Our results:} For additive functions, it is easy to show that dimension of $\ftwo$-sketches is $O(\min(\|w\|_1^2/\eps, n))$ (Corollary~\ref{cor:additive}) and give a matching communication lower bound (Theorem~\ref{thm:weighted-linear}) for all $\epsilon \ge \|w\|_2^2$.

\item \textbf{Budget-additive.} $f(x) = \min(b, \sum_{i = 1}^n w_i x_i)$ where $b, w_i \in \mathbb R$.
An example of such functions is the ``hockey stick'' function $hs_{\alpha} (x) = \min(\alpha, \frac{2\alpha}{n}\sum_{i = 1}^n x_i)$.

\textbf{Our results:} For budget-additive functions, it is easy to show that dimension of $\ftwo$-sketches is $O(\min(\|w\|_1^2/\eps, n))$ (Corollary~\ref{cor:budget-additive}). We give a matching communication bound for the ``hockey stick'' function for constant $\eps$ (Theorem~\ref{thm:lb-hockey-stick-uniform}) which holds even under the uniform distribution of the input.

\item \textbf{Coverage.} A function $f$ is a \textit{coverage function} on some universe $U$ of size $m$ if there exists a collection $A_1, \dots, A_n$ of subsets of $U$ and a vector of non-negative weights $(w_1, \dots, w_m)$ such that: 
$$f(S) = \sum_{i \in \cup_{j \in S} A_j} w_i.$$ 

\textbf{Our results:} We show a simple upper bound of $O(1/\eps)$ (Corollary~\ref{cor:coverage}) for such functions.

\item \textbf{Matroid rank.} A pair $M = ([n], \cI)$ is called a matroid if $\cI \subseteq 2^{[n]}$ is a non-empty set family such that the following two properties are satisfied:
\begin{itemize}
\item If $I \in \cI$ and $J \subseteq I$, then $J \in \cI$
\item If $I, J \in \cI$ and $|J| < |I|$, then there exists an $i \in I \setminus J$ such that $J \cup \{i\} \in \cI$.
\end{itemize}
The sets in $\cI$ are called \textit{independent}. A maximal independent set is called a \textit{base} of $M$. All bases have the same size, which is called the \textit{rank} of the matroid and is denoted as $rk(M)$. The \textit{rank function} of the matroid is the function $rank_M \colon 2^{[n]} \to \mathbb N_+$ defined as:
$$rank_M(S) := \max\{|I| \colon I \subseteq S, I \in \cI\}.$$
It follows from the definition that $rank_M$ is always a submodular $1$-Lipschitz function.

\textbf{Our results:}
In order to have consistent notation with the rest of the manuscript we always assume that matroid rank functions are scaled so that their values are in $[0,1]$.
Some of our results are exact, i.e. the corresponding matroid rank function is computed exactly (and in this case rescaling does not matter) while others allow approximation of the function value. In the latter case, the approximation guarantees are multiplicative with respect to the rescaled function.

Our main theorem regarding sketching of matroid rank functions is as follows:
\begin{theorem}[Sketching matroid rank functions]
For (scaled) matroid rank functions:
\begin{itemize}
\item
There exists an exact $\ftwo$-sketch of size $O(1)$ for matroids of rank $2$ (Theorem~\ref{thm:rank-2-matroids}) and graphic matroids (Theorem~\ref{thm:graphic-matroids}).
\item
There exists $c=\Omega(1)$ and a matroid of rank $r$ such that a $c$-approximation of its matroid rank function has randomized linear sketch complexity $\Omega(r)$. Furthermore, this lower bound also holds for the corresponding one-way communication problem (Theorem~\ref{thm:xor:lb:lipschitz}, Corollary~\ref{cor:lb:rank}).
\end{itemize}
\end{theorem}

This can be contrasted with the results under the uniform distribution for which matroids of rank $r$ have an exact $\ftwo$-sketch of size $O\left(\left(r\log r+\log\frac{1}{\eps}\right)^{r+1}\right)$, where $\eps$ is the probability of failure (Corollary~\ref{thm:uniform:matroid:exact}, follows from the junta approximation of~\cite{BOSY13}). Furthermore, matroids of high rank $\Omega(n)$ can be trivially approximately sketched under product distributions, due to their concentration around their expectation (see Appendix~\ref{app:uniform-sketches} for details).


\item \textbf{Lipschitz submodular.} A function $f \colon 2^{[n]} \to \mathbb R$ is $\alpha$-Lipschitz submodular if it is both submodular and $\alpha$-Lipschitz.

\textbf{Our results:} We show an $\Omega(n)$ communication lower bound (and hence a lower bound on $\ftwo$-sketch complexity) for constant error for monotone non-negative $O(1/n)$-Lipschitz submodular functions (Theorem~\ref{thm:lb-lipschitz-submodular}). We note that this hardness result crucially uses a non-product distribution over the input variables since Lipschitz submodular functions are tightly concentrated around their expectation under product distributions (see e.g.~\cite{V10,BH10}) and hence can be approximated using their expectation without any sketching at all.

\end{itemize}

\subsection{Overview and Techniques}
\subsubsection{Basic Tools: XOR Functions, Spectral Norm, Approximate Fourier Dimension}
In Section~\ref{sec:f2-sketching}, we introduce the basics of approximate $\ftwo$-sketching. 
Most definitions and results in this section can be seen as appropriate generalizations regarding Boolean functions (such as in~\cite{KMSY18}) to the case of real-valued functions where we replace Hamming distance with expected squared distance.
 We then define the randomized one-way communication complexity of the two-player XOR-function $\fplus{}(x,y) = f(x + y)$ corresponding to $f$.
 This communication problem plays an important role in our arguments as it gives a lower bound on the sketching complexity of $f$. 
 We then introduce the notion of approximate Fourier dimension developed in~\cite{KMSY18}. 
The key structural results of~\cite{KMSY18}, which characterize both the sketching complexity of $f$ and the one-way communication complexity of $\fplus{}$  under the uniform distribution using the approximate Fourier dimension, can be extended to the real-valued case as shown in Proposition~\ref{prop:uniform-approx-fourier-dimension} and Theorem~\ref{thm:approx-f2-sketch-uniform}.
This characterization is our main tool for showing lower bounds under the uniform distribution of $x$.

Another useful basic tool is a bound on the linear sketching complexity based on the spectral norm of $f$ which we develop in Appendix~\ref{app:fourier}.
In particular, as we show in Appendix~\ref{sec:fourier-sampling}, analogously to the Boolean case, we can leverage properties of the Fourier coefficients of a function $f$ to show that the $\eps$-approximate randomized sketching complexity of $f$ is at most $O(\|\hat f\|_1^2/\epsilon)$. 
Thus, we can determine the dimension of $\ftwo$-sketches for classes of functions whose spectral norms are well-bounded as well as functions which can be computed as Lipschitz compositions of a small number of functions with bounded spectral norm (Proposition~\ref{prop:composition}). Examples of such classes include additive (linear), budget-additive and coverage functions. 
Finally, we  argue that the dependence on the parameters in the spectral norm bound cannot be substantially improved in the real-valued case by presenting a subclass of linear functions which require sketches of size $\Omega(\|\hat f\|_1^2/\epsilon)$ (Theorem~\ref{thm:weighted-linear}). This is in contrast with the case of Boolean functions studied in~\cite{KMSY18} for which such tightness result is not known.

\subsubsection{Matroid Rank Functions, LTF, LTF$\circ$OR}
In Section~\ref{sec:matroid}, we present our results on sketching matroid rank and Lipschitz submodular functions. 
In Section~\ref{sec:rank-2} we show that matroid rank functions of matroids of rank $2$ and graphic matroids  have constant randomized sketching complexity. This is done by first observing that rank functions of such matroids can be expressed as a threshold function over a number of disjunctions. 
Therefore, it remains to determine the sketching complexity of the threshold function on a collection of disjunctions. 
Unfortunately, known upper bounds for the sketching complexity of even the simpler class of linear threshold functions have a dependence on $n$ and hence one cannot get a constant upper bound directly. 

Hence we show how to remove this dependence in Section~\ref{sec:ltf}, also resolving an open question of Montanaro and Osborne \cite{MO09}.
Recall that a linear threshold function (LTF) can be represented as $f(x)=\sgn{\sum_{i=1}^n w_i x_i -\theta}$ for some weights $w_i$ and threshold $\theta$, where we slightly alter the traditional definition of the sign function $\mathsf{sgn}$ to output $0$ if the input is negative and $1$ otherwise.  
An important parameter of an LTF is its \textit{margin} $m$, which corresponds to the difference between the threshold and the value of the linear combination closest to it. 
We first observe that the terms with insignificant coefficients, i.e. weights that are small in absolute value, do not contribute to the final output and thus, we can ignore them. 
Similarly, the remaining weights can be rounded, without altering the output of the function, to a collection of weights whose size is bounded, independent of $n$.  
Furthermore, $f(x)=0$ only if $x_i=1$ for at most $\frac{\theta}{2m}$ of these ``significant'' indices $i$ of $x$.
Thus, we hash the significant indices to a large, but independent of $n$, number of buckets.
As a result, either there are a small number of significant indices $i$ with $x_1=1$ and there are no collisions, or there is a large number of significant indices $i$ with $x_i=1$.
Since we can differentiate between these two cases, the sketch can output whether $f(x)=0$ or $f(x)=1$ with constant probability. 
With a more careful choice of hash functions this idea can be extended to linear thresholds of disjunctions. We show in Section~\ref{sec:lt:disj} that a threshold function over a number of disjunctions (LTF$\circ$OR) also has linear sketch complexity independent of $n$. 

In Section~\ref{sec:lin-sketch-lipschitz}, we show that there exists an $\Omega(n)$-Lipschitz submodular function $f$ that requires a randomized linear sketch of size $\Omega(n)$. 
We construct such a function probabilistically by using a large family of matroid rank functions constructed by~\cite{BH10} with an appropriately chosen set of parameters. 
We show any fixed deterministic sketch fails on a matroid chosen uniformly at random from this parametric family with very high probability. 
In fact, even if we take a union bound over all possible sketches of bounded dimension, the failure of probability is still negligibly close to $1$. 
By Yao's principle, the randomized linear sketch complexity follows.  
We then extend this result to a communication lower bound for $\fplus{}$ in Section~\ref{sec:one-way-lipschitz}.
In the one-way communication complexity setting, we show that there exists an $\Omega(n)$-Lipschitz submodular function $f$ whose $\fplus{}$ requires communication $\Omega(n)$.

\subsubsection{Uniform Distribution}
In Section~\ref{sec:uniform:hs:cc}, we show lower bounds for a budget additive ``hockey stick'' function under the uniform distribution. 
The lower bounds follow from a characterization of communication complexity using approximate Fourier dimension, and to complete the analysis, we lower bound the Fourier spectrum of the hockey stick function in Appendix~\ref{app:hs}. 
Although our approach for matroids of rank $2$ does not seem to immediately generalize to matroids of higher rank under arbitrary distributions, we show in Section~\ref{app:uniform-sketches} that under the uniform distribution, we can use $\eps$-approximations of disjunctive normal forms (DNFs) by juntas to obtain a randomized linear sketch whose size is independent of $n$. 
Furthermore, rank functions of matroids of very high rank admit trivial \emph{approximate} sketches under the uniform distribution as follows from standard concentration results~\cite{V10} (see Appendix~\ref{app:uniform-sketches}). 

%% file: prelims.tex
\section{Basics of Approximate $\ftwo$-Sketching}
\label{sec:f2-sketching}
\subsection{Communication Complexity of XOR functions}
In order to analyze the optimal dimension of $\ftwo$-sketches, we need to introduce a closely related  communication complexity problem.
For $f \colon \ftwo^n \to \mathbb R$ define the XOR-function $\fplus{}\colon \ftwo^n \times \ftwo^n \to \mathbb R$ as $\fplus{}(x,y) = f(x + y)$ where $x,y \in \ftwo^n$.
Consider a communication game between two players Alice and Bob holding inputs $x$ and $y$ respectively.
Given access to a shared source of random bits Alice has to send a single message to Bob so that he can compute $\fplus{}(x,y)$.
This is known as the one-way communication complexity problem for XOR-functions (see~\cite{SZ08,ZS10,MO09,LZ10,LLZ11,SW12,LZ13,TWXZ13,L14,HHL16,KMSY18} for related communication complexity results). 
\begin{definition}[Randomized one-way communication complexity of XOR function]
	\label{def:rand-oneway}
	For a function $f \colon \ftwo^n \to \mathbb R$, the \emph{randomized one-way communication complexity} with error $\delta$ (denoted as $\rc{\delta}(\fplus{})$) of its XOR-function is defined as the smallest size\footnote{Formally the minimum here is taken over all possible protocols where for each protocol the size of the message $M(x)$ refers to the largest size (in bits) of such message taken over all inputs $x \in \ftwo^n$. See~\cite{KN97} for a formal definition.} (in bits) of the (randomized using public randomness) message $M(x)$ from Alice to Bob, which allows Bob to evaluate $\fplus{}(x,y)$ for any $x,y \in \ftwo^n$ with error probability at most $\delta$.
\end{definition}

It is easy to see that $\rc{\delta}(\fplus{}) \le \rl{\delta}(f)$ as using shared randomness Alice can just send $k$ bits $\chi_{\mathbf{S}_1}(x),\chi_{\mathbf{S}_2}(x),\ldots, \chi_{\mathbf{S}_k}(x)$ to Bob, who can for each $i \in [k]$ compute $\chi_{\mathbf{S}_i}(x + y) = \chi_{\mathbf{S}_i}(x) + \chi_{\mathbf{S}_i}(y)$, which is an $\ftwo$-sketch of $f$ on $x + y$ and hence suffices for computing $\fplus{}(x,y)$ with probability $1 - \delta$. 

Replacing the guarantee of exactness of the output in the above definition with an upper bound on expected squared error, we obtain the following definition.   
\begin{definition}[Randomized one-way communication complexity of approximating an XOR function]
	\label{def:rand-oneway:approx}
	For a function $f \colon \ftwo^n \to \mathbb R$, the \emph{randomized one-way communication complexity}  (denoted as $\rca{\eps}(\fplus{})$) of approximating its XOR-function with error $\epsilon$ is defined as the smallest size(in bits) of the (randomized using public randomness) message $M(x)$ from Alice to Bob, which allows Bob to evaluate $\fplus{}(x,y)$ for any $x,y \in \ftwo^n$ with expected squared error at most $\epsilon$.
\end{definition}
Distributional communication complexity is defined analogously for the corresponding XOR function and is denoted as $\cD_\eps^{}$.

Finally, in the simultaneous model of computation \cite{BK97,BGKL03}, also called simultaneous message passing (SMP) model, there exist two players and a coordinator, who are all aware of a function $f$. 
The two players maintain $x$ and $y$ respectively, and must send messages of minimal size to the coordinator so that the coordinator can compute $f(x\oplus y)$.
\begin{definition}[Simultaneous communication complexity of XOR function]
\label{def:simul}
For a function $f \colon \ftwo^n \to \mathbb R$, the \emph{simultaneous one-way communication complexity} with error $\delta$ (denoted as $R_{\delta}^{sim}(f^+)$) of its XOR-function is defined as the smallest sum of the sizes (in bits) of the (randomized using public randomness) messages $M(x)$ and $M(y)$ from Alice and Bob, respectively, to a coordinator, which allows the coordinator to evaluate $\fplus{}(x,y)$ for any $x,y \in \ftwo^n$ with error probability at most $\delta$.
\end{definition}
Observe that a protocol for randomized one-way communication complexity of XOR function translates to a protocol for the simultaneous model of computation.

\subsection{Distributional Approximate $\ftwo$-Sketch Complexity}
Fourier analysis plays an important role in the analysis of distributional $\ftwo$-sketch complexity over the uniform distribution. 
In our discussion below, we make use of some standard facts from Fourier analysis of functions over $\ftwo^n$.
For definitions and basics of Fourier analysis of functions of such functions we refer the reader to the standard text~\cite{OD14} and Appendix~\ref{app:fourier}.
In particular, Fourier concentration on a low-dimensional subspace implies existence of a small sketch which satisfies this guarantee: 
\begin{definition}[Fourier concentration]
A function $f \colon \ftwo^n \rightarrow \mathbb R$  is $\gamma$-concentrated on a linear subspace $A_d$ of dimension $d$ if for this subspace it satisfies:
$$\sum_{S \in A_d} \hat f(S)^2 \ge \gamma.$$
\end{definition}

We also use the following definition of approximate Fourier dimension from~\cite{KMSY18}, adapted for the case of real-valued functions.

\begin{definition}[Approximate Fourier dimension]\label{def:approx-fourier-dim}
	Let $\mathcal A_k$ be the set of all linear subspaces of $\mathbb F_2^n$ of dimension $k$.
	For $f \colon \ftwo^n \to \mathbb R$ and $\epsilon \in (0,\|f\|_2^2]$ the $\epsilon$-approximate Fourier dimension $\dim_{\epsilon}(f)$ is defined as:
	$$\dim_\epsilon(f) = \min_k \left\{ \exists A \in \mathcal A_k \colon \sum_{\alpha \in A} \hat f^2(\alpha) \ge \epsilon\right\}.$$
\end{definition}
\begin{proposition}\label{prop:uniform-approx-fourier-dimension}
For any $f \colon \ftwo^n \rightarrow \mathbb R$, it holds that:
$$\rlad{U}{\eps}(f) \le \dim_{\|f\|_2^2 - \eps}(f).$$
\end{proposition}
\begin{proof}
Indeed, let $A_d$ be a $d$-dimensional subspace such that $\sum_{S \in A_d} \hat f^2(S) \ge \|f\|_2^2 - \eps$ and consider the function $g(x) = \sum_{S \in A_d} \hat f(S) \chi_S(x)$.
Note that in order to compute all values $\chi_S(x)$ for $S \in A_d$ it suffices to evaluate $d$ parities corresponding to sets $S_1, \dots, S_d$ forming a basis in $A_d$. Values of all other parities can be computed as linear combinations.
Let $\Delta(x) = f(x) - g(x)$.
Then the desired guarantee follows from the following calculation: 
\[\E_{x \sim U(\{0,1\}^n)}[\Delta(x)^2] = \E_{S \sim U(\{0,1\}^n)}[\hat{\Delta}(S)^2] = \sum_{S \in \{0,1\}^n} (\hat f(S) - \hat g(S))^2 = \sum_{S \notin A_d} \hat f(S)^2 \le \epsilon,\]
where the first equality holds from Parseval's identity.
\end{proof}

Furthermore, approximate Fourier dimension can be used as a lower bound on the one-way communication complexity of the corresponding XOR-function. We defer the proof of the following result to Appendix~\ref{app:thm:approx-f2-sketch-uniform} as it is follows closely an analogous result for Boolean functions from~\cite{KMSY18}.
\begin{theorem}\label{thm:approx-f2-sketch-uniform}
	For any $f \colon \ftwo^n \to \mathbb R$,  $\delta \in [0,1/2]$ and $\xi = \|f\|_2^2 - \eps(1 + 2\delta)$ it holds that:  $$ \distcu{\epsilon}(f^+) \ge \frac{\delta}{2} \cdot \dim_{\xi}(f).$$
\end{theorem}

%% file: matroid.tex
\section{Sketching Matroid Rank Functions}
\label{sec:matroid}
In this section we analyze sketching complexity of matroid rank functions. 
We start by considering the most fundamental possible matroids (of rank $2$) in Section~\ref{sec:rank-2} and showing that exactly sketching the matroid rank function requires $O(1)$ complexity. 
Similarly, we show that exactly sketching the rank of graphic matroids only uses $O(1)$ complexity. 
On the other hand, we show a lower bound in Section~\ref{sec:lin-sketch-lipschitz} that even approximating the rank $r$ of general matroids up to certain constant factors requires $\Omega(r)$ complexity.

To sketch matroids of rank $2$, we leverage a result by Acketa \cite{A78} which characterizes the collection of independent sets of such matroids. 
This allows us to represent matroid rank functions for matroids of rank $2$ as a linear threshold of disjunctions. 
Thus, we first show the randomized linear sketch complexity of $(\theta,m)$-linear threshold functions, resolving an open question by Montanaro and Osborne \cite{MO09}. 

\subsection{Matroids of Rank 2 and Graphic Matroids}\label{sec:rank-2}
In this section, we show that there exists a constant-size sketch that can be used to compute exact values of matroid rank functions for matroids of rank $2$.
\begin{theorem}\label{thm:rank-2-matroids}
For every matroid $M$ of rank $2$ it holds that $\rl{\frac13}(rank_M) = O(1)$.
\end{theorem}
\noindent
It is well-known that matroids of rank $2$ admit the following characterization (see e.g.~\cite{A78}). 
\begin{fact}\label{fact:rank-2-matroid}
The collection of size $2$ independent sets of a rank $2$ matroid can be represented as the edges in a complete graph that has edges of some number of disjoint cliques removed.
\end{fact}
\noindent
We define the following function as a threshold on the Hamming weight of a binary vector $x$
\[\mathsf{HAM}_{\le d}(x)=\begin{cases}
0,\qquad&\text{if }\sum_{i=1}^n x_i\le d+\frac{1}{2}\\
1,\qquad&\text{otherwise}.
\end{cases}
\]
We use a series of technical lemmas in the following section to prove the following result, which says that linear threshold functions can be succinctly summarized:
\begin{theorem}
\label{thm:ham:sketch}
The function $\mathsf{HAM}_{\le d}\left(\bigvee_{i\in S_1}x_i,\bigvee_{i\in S_2}x_i,\ldots\right)$ has a randomized linear sketch of size $O(d^2\log d)$.
\end{theorem}
The following fact that upper bounds the sketch complexity for functions with small support:
\begin{fact}[Folklore, see e.g.~\cite{MO09, KMSY18}]
\label{fact:small:support}
For any function $f:\{0,1\}^n\to\{0,1\}$ with $\min_{z\in\{0,1\}}\mathbf{Pr}_{x\in\{0,1\}^n}(f(x)=z)\le\eps$ it holds that 
$R_{\delta}^{lin}(f)\le\log\frac{2^{n+1}\eps}{\delta}$.
\end{fact}
Using Fact~\ref{fact:rank-2-matroid}, Theorem~\ref{thm:ham:sketch}, and Fact~\ref{fact:small:support}, we prove Theorem~\ref{thm:rank-2-matroids} by writing the matroid rank function for $M$ as a linear threshold function of disjunctions. 
\begin{proofof}{Theorem~\ref{thm:rank-2-matroids}}
We first claim $\ftwo$-sketching complexity of the rank function of any rank $2$ matroid $M$ is essentially the same as the complexity of the corresponding Boolean function that takes value $1$ if $rank_M(x) = 2$ and takes value $0$ otherwise.
Indeed, let the function above be denoted as $f_M$. Without loss of generality, we can assume that all singletons are independent sets in $M$ as otherwise the rank function of $M$ does not depend on the corresponding input.
Hence $rank_M(x) = 0$ if and only if $x = 0^n$. Thus $\rl{\delta}(rank_M) = \rl{\delta}(f_M) + O(\log 1/\delta)$ as by Fact~\ref{fact:small:support} we can use $O(\log 1/\delta)$-bit sketch to check whether $x = 0^n$ first and then evaluate $rank_M$ using $f_M$.
Recall from Fact~\ref{fact:rank-2-matroid} that matroids of rank $2$ can be represented as edges in a complete graph with edges corresponding to some disjoint union of cliques removed.

Let $S_1, \dots, S_t$ be the collection of vertex sets of disjoint cliques defining a rank $2$ matroid $M$ in Fact~\ref{fact:rank-2-matroid}.
Without loss of generality, we can assume that $|\cup_{i = 1}^t S_i| = n$ by adding singletons. 
Then: 
\[f_M(x) = \mathsf{HAM}_{\ge 2}\left(\bigvee_{j \in S_1} x_j,  \bigvee_{j \in S_2} x_j, \dots, \bigvee_{j \in S_t} x_j\right),\]
where $\mathsf{HAM}_{\ge 2}(z_1, \dots, z_t) = 1$ if and only if $\sum_{i = 1}^t z_i \ge 2$ is the threshold Hamming weight function.
By Theorem~\ref{thm:ham:sketch}, the sketch complexity of $f_M(x)$ is $O(1)$, since the Hamming weight threshold is $d=2$.

\end{proofof}
Since the independent bases of a graphic matroid $M(G)$ are the spanning forests of $G$, the matroid rank function of a graphic matroid of rank $r$ can be expressed as
\[f_M(x) = \mathsf{HAM}_{\ge r}\left(\bigvee_{j \in S_1} x_j,  \bigvee_{j \in S_2} x_j, \dots, \bigvee_{j \in S_t} x_j\right),\]
where each $S_i$ is a separate spanning forest. 
Therefore, Theorem~\ref{thm:ham:sketch} yields a $O(r^2\log r)$ space linear sketch for graphic matroids of rank $r$.
\begin{theorem}
\label{thm:graphic-matroids}
For every graphic matroid $M$ of rank $r$, it holds that $\rl{\frac13}(rank_M) = O(r^2\log r)$.
\end{theorem}
We use the remainder of the Section~\ref{sec:rank-2} to prove Theorem~\ref{thm:ham:sketch}, while resolving an open question by Montanaro and Osborne \cite{MO09}. 

\subsubsection{Linear Threshold Functions}
\label{sec:ltf}
We first define linear threshold functions (LTFs) and $(\theta,m)$-LTFs.
\begin{definition}
A function $f:\{0,1\}^n\to\{0,1\}$ is a \emph{linear threshold function} (LTF) if there exist constants $\theta,w_1,w_2,\ldots,w_n$ such that $f(x)=\sgn{-\theta+\sum_{i=1}^n w_ix_i}$, where $\sgn{y}=0$ for $y<0$ and $\sgn{y}=1$ for $y\ge 0$ is the Heaviside step function. 
\end{definition}

\begin{definition}
A monotone linear threshold function $f:\{0,1\}^n\to\{0,1\}$ is a \emph{$(\theta,m)$-LTF} if $m\le\min_{x\in\{0,1\}^n}\left|-\theta+\sum_{i=1}^n w_ix_i\right|$, where $\theta$ is referred to as the \em{threshold} and $m$ as the \em{margin} of the LTF.
\end{definition}
Although $(\theta,m)$-LTFs have previously been shown to have randomized linear sketch complexity $O\left(\frac{\theta}{m}\log n\right)$ \cite{LZ13}, Montanaro and Osborne asked whether any $(\theta,m)$-LTF can be represented in the simultaneous model with $O\left(\frac{\theta}{m}\log\frac{\theta}{m}\right)$ communication. 
\begin{question}[\cite{MO09}]
\label{q:MO}
Let $g(x,y)=f(x\oplus y)$, where $f$ is a $(\theta,m)$-LTF. 
Does there exist a protocol for $g$ in the simultaneous model with communication complexity $O\left(\frac{\theta}{m}\log\frac{\theta}{m}\right)$?
\end{question}
\noindent
Note that the difference between $\log n$ and $\log\frac{\theta}{m}$ is crucial for obtaining constant randomized linear sketch complexity for functions for matroid rank $2$. 
We answer Question~\ref{q:MO} in the affirmitive and show the stronger result that $(\theta,m)$-LTFs admit a randomized linear sketch of size $O\left(\frac{\theta}{m}\log\frac{\theta}{m}\right)$. 
We first show that we can completely ignore all variables whose weights are significantly smaller than $2m$ in evaluating a $(\theta,m)$-LTF.

\begin{lemma}
\label{lem:ignore:weights}
Let $f(x)=\sgn{-\theta+\sum_{i=1}^n w_ix_i}$ be a $(\theta,m)$-LTF. 
For $1\le i\le n$, let $w'_i=w_i$ if $w_i\ge 2m$ and $w'_i=0$ otherwise.
Then $f(x)=\sgn{-\theta+\sum_{i=1}^n w'_ix_i}$.
\end{lemma}
\begin{proof}
We show the stronger result that for any $j$ such that $w_j<2m$, then $f(x)=f(x\oplus e_j)$, where $e_j$ is the elementary unit vector with one in the $j\th$ position, and zeros elsewhere. 
This implies the lemma since it shows that any variable whose weight is less than $2m$ does not affect the output of the function or the margin of the function and thus might as well have weight zero.

Suppose, by way of contradiction, that $f(x)\neq f(x\oplus e_j)$ and without loss of generality, $f(x)=0$ with $x_j=0$. 
Since $f$ is a linear threshold function and $f(x)=0$, then $-\theta+\sum_{i=1}^n w_ix_i<0$. 
Moreover, $f$ is a $(\theta,m)$-LTF, so $-\theta+\sum_{i=1}^n w_ix_i<-m$. 
Because $w_j<2m$, $-\theta+w_j+\sum_{i=1}^n w_ix_i<-\theta+2m+\sum_{i=1}^n w_ix_i<m$. 
But because $m$ is the margin of the function, if $-\theta+w_j+\sum_{i=1}^n w_ix_i<m$, then it must hold that $-\theta+w_j+\sum_{i=1}^n w_ix_i<-m$.
Therefore, $f(x\oplus e_j)=0$, so $x_j$ does not affect the output of the function or the margin of the function. 
\end{proof}
As noted, Lemma \ref{lem:ignore:weights} implies that we can ignore not only variables with zero weights, but all variables whose weights are less than $2m$. 
We now bound the support of the set $\{x\ |\ f(x)=0\}$, where $f$ is a $(\theta,m)$-LTF, and apply Fact~\ref{fact:small:support}. 
\begin{lemma}
\label{lem:sketch}
For any $(\theta,m)$-LTF, there exists a randomized linear sketch of size $O\left(\frac{\theta}{m}\log n\right)$.
\end{lemma}
\begin{proof}
Let $f(x)=\sgn{-\theta+\sum_{i=1}^n w_ix_i}$ be a $(\theta,m)$-LTF. 
By Lemma \ref{lem:ignore:weights}, the output of $f$ remains the same even if we only consider the variables $S$ with weight at least $2m$. 
On the other hand, if $f(x)=0$, then at most $\frac{\theta}{2m}$ variables in $S$ can have value $1$. 
Equivalently, at most $\frac{\theta}{2m}$ indices $i$ can have $x_i=1$ if $f(x)=0$. 
Thus, the number of $x\in\{0,1\}^n$ with $f(x)=0$ is at most $\sum_{0\le i\le\theta/2m}\binom{n}{i}\le(n+1)^{\lceil\theta/2m\rceil}$.
Applying Fact \ref{fact:small:support}, there exists a randomized linear sketch for $f$, of size $O\left(\frac{\theta}{m}\log n\right)$.
\end{proof}

In order to fully prove Question~\ref{q:MO} and obtain a dependence on $\log\frac{\theta}{m}$ rather than $\log n$, we use the following two observations. 
First, we show in Lemma \ref{lem:round:weights} that the weights of a $(\theta,m)$-LTF can be rounded to a set that contains $O\left(\frac{\theta}{m}\right)$ elements. 
Second, we show in Theorem \ref{thm:sketch} that we can then use hashing to reduce the number of variables down to $\mathsf{poly}\left(\frac{\theta}{m}\right)$ before applying Lemma \ref{lem:sketch}.

\begin{lemma}
\label{lem:round:weights}
Let $f(x)=\sgn{-\theta+\sum_{i=1}^n w_ix_i}$ be a $(\theta,m)$-LTF. 
Then there exists a set $W$ with $|W|=O\left(\frac{\theta}{m}\log\frac{\theta}{m}\right)$, and a margin $m'=\Theta(m)$ such that $f(x)=\sgn{-\theta+\sum_{i=1}^n w'_ix_i}$, where each $w'_i\in W$ and $f$ is a $(\theta,m')$-LTF. 
\end{lemma}
\begin{proof}
Observe that for any $w_i\ge2\theta$, if $x_i=1$, then $f(x)=1$. 
Thus, if $f(x)=1$, it suffices to consider $2m\le w_i\le 2\theta$. 

Let $W=\{2m(1+\eps)^i\}_{i=0}^t$ for $t=\left\lceil\log_{1+\eps}\left(\frac{\theta}{m}\right)\right\rceil$, where $\eps$ is some fixed constant that we set at a later time.  
For each $i$, let $w'_i$ be the largest element in $W$ that does not exceed $w_i$. 
Thus, $w'_i\le w_i<(1+\eps)w'_i$. 
Observe that since $w'_i\le w_i$ and $f$ is a $(\theta,m)$-LTF, then $f(x)=0$ implies $-m>-\theta+\sum_{i=1}^n w_ix_i\ge-\theta+\sum_{i=1}^n w_ix_i$, so that $\sgn{-\theta+\sum_{i=1}^n w'_ix_i}=0=f(x)$ and a margin of $m$ remains.

On the other hand, if $f(x)=1$, then $\sum_{i=1}^n w_ix_i>\theta+m$ as $f$ is a $(\theta,m)$-LTF. 
Since $w'_i\le w_i<(1+\eps)w'_i$, then $\sum_{i=1}^n w'_ix_i>\frac{\theta+m}{1+\eps}>(1-\eps)(\theta+m)$.
Observe that $\theta\ge m$ and hence, $\sum_{i=1}^n w'_ix_i>\theta-\eps\theta+m-\eps m\ge\theta+m-2\eps\theta$.
Setting $\eps=\frac{\theta}{10m}$ shows that $\sgn{-\theta+\sum_{i=1}^n w'_ix_i}=1=f(x)$ and a margin of $m'=\frac{4}{5}m$ remains.
\end{proof}

\noindent
The following result is also useful for our construction of a sketch for a $(\theta,m)$-LTF.

\begin{lemma}\cite{HSZZ06}
\label{lem:ham:sketch}
There is a randomized linear sketch with size $O(1)$ for the function
\[\mathsf{HAM}_{n,d|2d}(x)=\begin{cases}
1,\qquad&\text{if }||x||_0\le d\\
0,\qquad&\text{if }||x||_0\ge 2d
\end{cases}
\]
on instances $\{x|x\in\{0,1\}^n\text{ and }||x||_0\le d\text{ or }||x||_0\ge 2d\}$.
\end{lemma}

\begin{fact}
\label{fact:collision}
If $h:[n]\to[M]$ is a random hash function and $S\subseteq[n]$, then the probability that there exist $x,y\in S$ with $h(x)=h(y)$ is at most $\frac{|S|^2}{M}$.
\end{fact}

\begin{theorem}
\label{thm:sketch}
Any $(\theta,m)$-LTF admits a randomized linear sketch of size $O\left(\frac{\theta}{m}\log\frac{\theta}{m}\right)$.
\end{theorem}
\begin{proof}
Let $f(x)=\sgn{-\theta+\sum_{i=1}^n w_ix_i}$ be a $(\theta,m)$-LTF. 
By Lemma \ref{lem:round:weights}, we can assume that $w_i\in W=\{2m(1+\eps)^i\}_{i=0}^t$ so that the new margin $m'=\frac{4}{5}m$ and $t=\left\lceil\log_{1+\eps}\frac{\theta}{m}\right\rceil$ for $\eps=\frac{\theta}{10m}$. 
Recall from Lemma \ref{lem:sketch}, $f(x)=0$ only if $x_i=1$ for at most $\frac{\theta}{2m}$ indices $i$ of $x$. 
From Lemma \ref{lem:ham:sketch}, we can detect the instances where at least $\frac{\theta}{2m}$ indices $i$ of $x$ satisfy $x_i=1$. 

On the other hand, if less than $\frac{\theta}{2m}$ indices $i$ of $x$ satisfy $x_i=1$, we can identify these indices and corresponding weights via hashing. 
Let $h:[n]\to[M]$, where $M=5\left(\frac{\theta}{m}\right)^2$, and $S$ be a set of indices of $x$, of size at most $\frac{\theta}{m}$. 
Then by Fact \ref{fact:collision}, the probability of a collision in $h$ under elements of $S$ is at most $\frac{1}{5}$.
We partition $[n]$ into sets $S_{w,j}$ where $w\in W$ and $j\in[M]$ so that $S_{w,j}=\{i|h(i)=j\wedge w_i=w\}$. 
Therefore with probability at least $\frac{4}{5}$, there are no collisions in $h$ under elements of $S$ and $|S_{w,j}|\le 1$ for all $w\in W$ and $j\in[M]$. 

Let $y_{w,j}=\sum_{i\in S_{w,j}}x_i$ and note that if there are no collisions in $h$ under elements of $S$, then
\[\sum_{i=1}^n w_ix_i=\sum_{(j,w)\in[M]\times W} w\left(\sum_{i\in S_{w,j}} x_i\right)=\sum_{(j,w)\in[M]\times W}w\cdot y_{w,j}.\]
Thus, $f(x)$ is equivalent to the function $g(y)=\sgn{-\theta+\sum_{w,j}w\cdot y_{w,j}}$.
Since $|W|=O\left(\frac{\theta}{m}\log\frac{\theta}{m}\right)$, $M=5\left(\frac{\theta}{m}\right)^2$ and $m'=\frac{4}{5}m$ is the margin for $g(y)$, then $g(y)$ depends on $O\left(\left(\frac{\theta}{m}\right)^3\log\frac{\theta}{m}\right)$ variables $y_{w,j}$. 
By Lemma \ref{lem:sketch}, there exists a randomized sketch for $g(y)$ of size $O\left(\frac{\theta}{m}\log\frac{\theta}{m}\right)$.
\end{proof}

\noindent
We can also show that Theorem \ref{thm:sketch} is tight by recalling the function 
\[\mathsf{HAM}_{\le d}(x)=\begin{cases}
0,\qquad&\text{if }\sum_{i=1}^n x_i\le d+\frac{1}{2}\\
1,\qquad&\text{otherwise}.
\end{cases}
\]
Since this function is a $\left(d+\frac{1}{2},\frac{1}{2}\right)$-LTF, it can be represented by a randomized linear sketch of size $O(d\log d)$. 
On the other hand, Dasgupta, Kumar and Sivakumar \cite{DKS12} notes that the one-way complexity of small set disjointness for two vectors $x$ and $y$ of weight $d$, which reduces to the function $\mathsf{HAM}_{\le d}(x\oplus y)$, is $\Omega(d\log d)$. 
Thus, $\mathsf{HAM}_{\le d}(x\oplus y)$ also requires a sketch of size $\Omega(d\log d)$.

\subsubsection{Linear Threshold of Disjunctions}
\label{sec:lt:disj}
In this section, we describe a randomized linear sketch for functions that can be represented as $2$-depth circuits where the top gate is a monotone linear threshold function with threshold $\theta$ and margin $m$, and the bottom gates are OR functions. 
Formally, if $\displaystyle g_S(x)=\bigvee_{i\in S} x_i$, $q$ is a linear threshold function, and $w_S\ge 0$, then $f(x)=q(\ldots,g_S(x),\ldots)=\sgn{-\theta+\sum_{S\in2^{[n]}}w_S\cdot g_S(x)}$.

\begin{lemma}
\label{lem:disjunction}
Let $f(x)=\sgn{-\theta+\sum_{i=1}^n w_ix_i}$ be a $(\theta,m)$-LTF where $w_i\in W$ for some set $W$. 
Let $h:[n]\to[M]$ be a random hash function where $M=\frac{50\theta^2}{m^2}$ and 
\[f_h(x)=\sgn{-\theta+\sum_{(j,w)\in[M]\times W} w\left(\bigvee_{\substack{i:h(i)=j\\w_i=w}}x_i\right)}.\] 
Then for all $x$, $\mathbf{Pr}\left[f_h(x)\neq f(x)\right]\le\frac{1}{50}$.
\end{lemma}
\begin{proof}
As by Lemma \ref{lem:round:weights}, we can assume without loss of generality that $w_i\ge 2m$ and $w\ge 2m$. 
Let $S=\{i|x_i=1\}$ so that if there are no collisions under $h$ in $S$, then
\[\sum_{(j,w)\in[M]\times W} w\left(\bigvee_{\substack{i:h(i)=j\\w_i=w}}x_i\right)=\sum_i w_ix_i.\]
If $f(x)=0$, then $|S|\le\frac{\theta}{2m}$ so that the probability there are collisions under $h$ in $S$ is at most $\frac{1}{200}$ by Fact \ref{fact:collision}. 
Thus if $f(x)=0$, then $f_h(x)=0$ with probability at least $1-\frac{1}{200}$.

If $f(x)=1$, then either $|S|<\frac{\theta}{m}$ or $|S|\ge\frac{\theta}{m}$. 
If $|S|<\frac{\theta}{m}$, then the probability there are collisions under $h$ in $S$ is at most $\frac{1}{50}$ by Fact \ref{fact:collision}, so then $f_h(x)=1$ with probability at least $1-\frac{1}{50}$. 
If $|S|\ge\frac{\theta}{m}$, with probability at least $1-\frac{1}{50}$, there exist $\frac{\theta}{m}$ values $j$ such that there exists $x_i=1$ and $h(i)=j$. 
Therefore, we set $f_h(x)=1$ whenever at least $\frac{\theta}{m}$ buckets of $h$ are non-empty.

In all cases, $f_h(x)=f(x)$ with probability at least $1-\frac{1}{50}$.
\end{proof}

\begin{theorem}
Let $g_S(x)=\bigvee_{i\in S} x_i$ with $w_S\ge 0$, $q$ be a $(\theta,m)$-LTF, and
\[f(x)=q(\ldots,g_S(x),\ldots)=\sgn{-\theta+\sum_{S\in2^{[n]}}w_S\cdot g_S(x)}.\]
Then there is a randomized linear sketch for $f$ of size $O\left(\left(\frac{\theta}{m}\right)^4\log^2\frac{\theta}{m}\right)$, where $m$ is the margin of $q$.
\end{theorem}
\begin{proof}
We first apply Lemma \ref{lem:ignore:weights} and Lemma \ref{lem:round:weights} to $q$ so that weights $w_i$ can be rounded to elements of a set $W$ with $|W|=O\left(\frac{\theta}{m}\log\frac{\theta}{m}\right)$. 
For each $w_i\in W$, it again suffices to detect whether $\Theta(\frac{\theta}{m})$ disjunctions are nonzero. 
Hence to hash $O\left(\left(\frac{\theta}{m}\right)^2\log\frac{\theta}{m}\right)$ disjunctions, it suffices to use a hash function with $M=O\left(\left(\frac{\theta}{m}\right)^4\log^2\frac{\theta}{m}\right)$ buckets.
By Lemma \ref{lem:disjunction}, our resulting randomized linear sketch has size $O\left(\left(\frac{\theta}{m}\right)^4\log^2\frac{\theta}{m}\right)$. 

\end{proof}
\begin{proofof}{Theorem~\ref{thm:ham:sketch}}
Recall that $\mathsf{HAM}_{\le d}(x)$ is a $\left(d+\frac{1}{2},\frac{1}{2}\right)$-LTF. 
Furthermore, the set of weights $W$ for $\mathsf{HAM}_{\le d}(x)$ consists of a single element $\{1\}$, since the coefficient of each disjunction is one. 
Since $M=O(d^2\log d)$, we can construct a randomized linear sketch with size $O(d^2\log d)$ by Lemma \ref{lem:disjunction}.
\end{proofof}
We note that our approach can be easily generalized to the case where the disjunction include the negations of some variables as well.

\subsection{Communication Complexity of Lipschitz Submodular Functions}
We discuss the communication complexity of Lipschitz submodular functions in this section. 
We first show in Section~\ref{sec:lin-sketch-lipschitz} that there exists an $\Omega(n)$-Lipschitz submodular function $f$ that requires a randomized linear sketch of size $\Omega(n)$. 
We then show in Section~\ref{sec:one-way-lipschitz} that in the one-way communication complexity model for XOR functions, there exists an $\Omega(n)$-Lipschitz submodular function $f$ that has communication complexity $\Omega(n)$. 

\subsubsection{Approximate $\ftwo$-Sketching of Lipschitz Submodular Functions}\label{sec:lin-sketch-lipschitz}
\begin{theorem}\label{thm:lb-lipschitz-submodular}
	There exist constants $c_1, c_2, \eps \ge 0$ and a monotone non-negative $(\frac{c_1}{n})$-Lipschitz submodular function $f$ (a scaling of a matroid rank function) such that:
	$$\rla{\eps}(f) \ge c_2 n.$$
\end{theorem}

\begin{proof}
Our proof uses a construction of a large family of matroid rank functions given in~\cite{BH10}, Theorem 8.
The construction uses the following notion of lossless bipartite expanders:
\begin{definition}[Lossless bipartite expander]
Let $G = (U \cup V, E)$ be a bipartite graph. For $J \subseteq U$ let $\Gamma(J) = \{v | \exists u \in U \colon \{u,v\} \in E\}$.
Graph $G$ is a $(D,L,\eps)$-lossless expander if:
\begin{align*}
|\Gamma(\{u\})| = D \quad & \forall u \in U \\
|\Gamma(J)| \ge (1 - \eps) D |J| \quad & \forall J \subseteq U, |J| \le L.
\end{align*}
\end{definition}

\noindent
Here we need different parameters than in~\cite{BH10} so we restate their theorem as follows:
\begin{theorem}[\cite{BH10}]\label{thm:bh-matroid}
	
	Let $(U \cup V, E)$ be a $(D,L,\epsilon)$-lossless expander with $|U| = k$ and $|V| = n$ and let $b = 8 \log k$.
	If $D \ge b$, $L = 4D/b - 2$ and $\eps = \frac{b}{4D}$ then there exists a family of sets $\cA \subseteq 2^{[n]}$ and a family of matroids $\{M_\cB \colon \cB \subseteq \cA\}$ with the following properties:
	\begin{itemize}
		\item $|\mathcal A| = k$ and for every $A \in \mathcal A$ it holds that $|A| = D$.
		\item For every $\cB \subseteq \cA$ and every $A \in \cA$, we have:
		\begin{align*}
		rank_{M_\cB}(A) = 
		\begin{cases}
		b &\text {\quad \quad \quad if } A \in \cB\\
		D &\text {\quad \quad \quad if } A \in \cA \setminus \cB
		\end{cases}
		\end{align*}
	\end{itemize}
\end{theorem}

\noindent
We use the following construction of lossless expanders from~\cite{V12}, see also~\cite{BH10}.

\begin{theorem}[\cite{V12}]
	Let $k \ge 2$ and $\eps \ge 0$. For any $L \le k$, let $D \ge 2 \log k/\epsilon$ and $n \ge 6 DL/\eps$. Then a $(D,L,\eps)$-lossless expander exists.
\end{theorem}

\noindent
In the above theorem we can set parameters as follows:
$$D = \frac{n}{3 \cdot 2^7}, \quad L = 2^3, \quad \eps = 2^{-3}, \quad k = 2^{n/3 \cdot 2^{11}}, \quad b = \frac{n}{3 \cdot 2^8}. $$
Note that under this choice of parameters we have $6 DL/\epsilon = n$ and $\frac{2 \log k}{ \epsilon} = D$ and hence a $(D,L,\eps)$-lossless expander with parameters set above exists. 


Now consider the family of matroids $\cM$ given by Theorem~\ref{thm:bh-matroid} using the expander construction above.
The rest of the proof uses the probabilistic method. We will show non-constructively that there exists a matroid in this family whose rank function does not admit a sketch of dimension $d = o(n)$.
Let $\cD = U(\cA)$ be the uniform distribution over $\cA$.
By Yao's principle it suffices to show that there exists a matroid rank function for which any deterministic sketch fails with a constant probability over this distribution.
In the proof below we first show that any fixed deterministic sketch succeeds on a randomly chosen matroid from $\cM$ with only a very tiny probability, probability $2^{2^{-\Omega(n)}}$, and then take a union bound over all $2^{dn}$ sketches of dimension at most $d$. 


Indeed, fix any deterministic sketch $\cS$ of dimension $d = n/2^{11}$.
Let $\{b_1, \dots, b_{2^d}\}$ be the set of all possible binary vectors of length $d$ corresponding to the possible values of the sketch, so that each $b_i \in \{0,1\}^d$. 

Let $S_{b_i} = \{A \in \cA : \cS(A) = b_i\}$. Let $t = \frac{1}{4} 2^{n/2^{11}}$ and $G = \{b_i \in \{0,1\}^d | |S_{b_i}| \ge t\}$. The following proposition follows by a simple calculation.

\begin{proposition}\label{prop:large-sets}
If $t = \frac{1}{4} 2^{n/2^{11}}$ then $\frac{1}{k} \sum_{b_i \in G} |S_{b_i}| \ge \frac34$. 
\end{proposition}
\begin{proof}
We have:
$$\frac{1}{k} \sum_{b_i \in G} |S_{b_i}| \ge 1 - \frac{1}{k} \sum_{b_i \colon |S_{b_i}| < \frac{k}{4 \cdot 2^d}} |S_{b_i}| \ge 1 - \frac{1}{k} \cdot \frac{k}{4 \cdot 2^d} \cdot 2^d \ge \frac34.\qedhere$$
\end{proof}

Let $S_{b_i}^1 = \{A \in S_{b_i} \colon rank_{M_\cB}(A) = b\}$ and $S_{b_i}^2 = \{A \in S_{b_i} \colon rank_{M_\cB}(A) = D\}$. 
We require the following lemma.
\begin{lemma}\label{lem:union-bound}
Let $t = \frac{1}{4} 2^{n/2^{11}}$ and $d = n/2^{11}$. There exists a matroid $M_{\cB} \in \cM$ such that for all deterministic sketches $\cS$ of dimension $d$ and all $b_i \in G$: 
$$\min(|S_{b_i}^1|, |S_{b_i}^2|) \ge \frac14 |S_{b_i}|.$$
\end{lemma}
\begin{proof}
The proof uses the probabilistic method to show the existence of $\cB$ with desired properties.
Consider drawing a random matroid from the family $\cM$, i.e. pick $\cB$ to be a uniformly random subset of $\cA$ and consider $M_\cB$.
Fix any deterministic sketch $\cS$ and any $b_i \in G$.
Since $|S_{b_i}| \ge t$, by the Chernoff bound, it holds that:
$$\Pr_{\cB \subseteq \cA}\left[\left|S_{b_i}^1\right| > \left(\frac12 + \delta \right)|S_{b_i}| \right] \le e^{- c \delta^2 |S_{b_i}|} \le e^{-c\delta^2 t}.$$
Setting $\delta = 1/4$, we have that the above probability is at most $e^{-Ct}$ for some constant $C > 0$.
Applying the argument above to both $S_{b_i}^1$ and $S_{b_i}^2$, we have that: 
$$\Pr_{\cB \subseteq \cA}\left[\min(\left|S_{b_i}^1\right|,\left|S_{b_i}^2\right|) <\frac14 |S_{b_i}| \right] \le 2e^{- C t}.$$
Let $\cE$ denote the event that $\min(\left|S_{b_i}^1\right|,\left|S_{b_i}^2\right|) \ge \frac14 |S_{b_i}|$.

Note that the total number of deterministic sketches of dimension $d$ is at most $2^{dn}$, since each sketch is specified by a collection of $d$ linear functions over $\mathbb F_2^n$. 
Also note that for each sketch $|G| \le 2^d$. 
Taking a union bound over all sketches and all sets $G$ by the choice of $t$ and $d$ event $\cE$ holds for all $\cS$ and $b_i \in G$ with probability at least: 
$$1 - 2^{(n + 1)d + 1}e^{-Ct} \ge 1 - 2^{(n + 1)d + 1} 2^{-\frac{C}{4} 2^{n/2^{11}}}  = 1 - o(1).$$
Thus, there exists some set $\cB$ for which the statement of the lemma holds.
\end{proof}

Fix the set $\cB$ constructed in Lemma~\ref{lem:union-bound} and consider the function $rank_{M_\cB}$.
Consider distribution $\cD$ over the inputs. The probability that any deterministic sketch over this distribution makes error at least $D - b$ is at least:
\begin{align*} 
\frac1k \sum_{b_i \in \{0,1\}^n} \min(|S_{b_i}^1|, |S_{b_i}^2|) &\ge \frac1k \sum_{b_i \in G} \min(|S_{b_i}^1|, |S_{b_i}^2|) \\
& \ge \frac1k \sum_{b_i \in G} \frac 14 |S_{b_i}| && \text{(by Lemma~\ref{lem:union-bound})}\\
 &\ge \frac34 \times \frac 14 \ge \frac16.&& \text{(by Proposition~\ref{prop:large-sets})}\\
\end{align*}



Finally, the construction of~\cite{BH10} ensures that the function $rank_{M_\cB}$ takes integer values between $0$ and $D$.
Using this and the fact that matroid rank functions are $1$-Lipschitz, we can normalize it by dividing all values by $D$ and ensure that the resulting function is $O(1/n)$-Lipschitz and takes values in $[0,1]$, while the sketch makes error at least $(D - b)/D = \frac12$. \qedhere
\end{proof}

\subsubsection{One-Way Communication of Lipschitz Submodular Functions}
\label{sec:one-way-lipschitz}

In this section, we strengthen the lower bound shown above, extending it to the corresponding one-way communication problem. 
We use the same notation as in the previous section.
\begin{theorem}
\label{thm:xor:lb:lipschitz}
There exists a constant $c_1> 0$ and a $\frac{c_1}{n}$-Lipschitz submodular function such that $R^{\rightarrow}_{1/3} = \Omega(n)$.
\end{theorem}
\begin{proof}
Let $\alpha=\frac{1}{3\cdot2^{11}}$ and $|\mathcal{A}|=k=2^{\alpha n}$. 
Suppose Alice holds $x\in\mathcal{A}\subseteq\{0,1\}^n$ and Bob holds $y\in\{0,1\}^n$. 
Recall that in the one-way communication model for XOR functions, Alice must pass a message of minimal length to Bob, who must then compute $f(x\oplus y)$ with some probability, say $\frac{2}{3}$. 
Here, we let specifically let $f$ be a scaling of a matroid rank function, which is some monotone non-negative $\left(\frac{c_1}{n}\right)$-Lipschitz submodular function. 
By Yao's principle, it suffices to show that every deterministic one-way communication protocol using at most $\frac{\alpha}{4}n$ bits fails with probability greater than $\frac{1}{3}$ over $\mathcal{A}$. 
Suppose by way of contradiction, that Alice and Bob succeed through a deterministic one-way communication protocol, using at most $\frac{\alpha}{4}n$ bits. 
For the purpose of analysis, we furthermore suppose that Bob's input is fixed. 

We now claim that if Alice passes a message to Bob using at most $\frac{\alpha}{4}n$ bits, then there are at least $2^{\alpha n}-4\cdot 2^{\alpha n/4}$ points in $\mathcal{A}$ that are represented by the same message as at least five other points.
Note that Alice can partition the input space $\mathcal{A}$ into at most $2^{\alpha n/4}$ parts, each part with its own distinct representative message. 
The number of points \emph{not} in parts containing at least five other points is at most $4\cdot 2^{\alpha n/4}$. 
The remaining points, at least $2^{\alpha n}-4\cdot 2^{\alpha n/4}$ in quantity, are represented by the same message as at least five other points.

Let $S$ be the set of points in $\mathcal{A}$ represented by a given message from Alice. 
Hence, Alice assigns the same message to each of these points and passes the state of the protocol to Bob. 
Because Bob cannot distinguish between these points and must perform a deterministic protocol, then Bob must output the same result for each of these points. 
Recall that we consider Bob's input $y\in\{0,1\}^n$ as fixed. 
Consider the family of functions 
\[\mathcal{F}=\{f: f(x\oplus y)=b\text{ or } f(x\oplus y)=D\text{ for all }x\in\mathcal{A}\}.\]
Thus, if $S$ contains at least five points, there exists $f\in\mathcal{F}$ such that Bob errs on at least $\frac{2}{5}$ fraction of the points in $S$ by setting $f(x\oplus y)=b$ to at least $\left\lfloor\frac{|S|-1}{2}\right\rfloor$ of the points $x\in S$ and similarly for $f(x\oplus y)=D$. 
Moreover, since Alice partitions the points in $\mathcal{A}$, then there exists an $f\in\mathcal{F}$ such that Bob errs on at least $\frac{2}{5}$ fraction on \emph{all} points that are represented by the same message as at least five other points. 
Hence, the total number of inputs that Bob errs is at least $\frac{2}{5}\left(2^{\alpha n}-6\cdot 2^{\alpha n/4}\right)>\frac{1}{3}\cdot 2^{\alpha n}$ for sufficiently large values of $n$. 
This contradicts the assumption that the communication protocol, using at most $\frac{\alpha}{4}n$ bits, succeeds with probability $\frac{2}{3}$.
\end{proof}
By restricting the $n$-dimensional elements to $r$ coordinates and observing that the construction outputs matroids of rank $b$ or $D$ that are separated by a constant gap, we obtain the following result using the same proof:
\begin{corollary}
\label{cor:lb:rank}
There exists $c=\Omega(1)$ such that a $c$-approximation of matroid rank functions has randomized one-way communication complexity $R^{\rightarrow}_{1/3} = \Omega(r)$ where $r$ is the rank of the underlying matroid.
\end{corollary}

%% file: appendix.tex
\section{Background}
\subsection{Fourier Analysis}\label{app:fourier}
We consider functions\footnote{\label{fourierrange}
	In all Fourier-analytic arguments Boolean functions are treated as functions of the form $f : \ftwo^n \to \oo$ where $0$ is mapped to $1$ and $1$ is mapped to $-1$. Otherwise we use these two notations interchangeably.} from $\ftwo^n$ to 
$\R$.
For any fixed $n \ge 1$, the space of these functions forms an inner product space
with the inner product
$\left<f, g\right> = \E_{x \in \ftwo^n}[ f(x) g(x) ] = \frac1{2^n} \sum_{x \in \ftwo^n} f(x)g(x)$.
The $\ell_2$ norm of $f : \ftwo^n \to \R$ is
$\| f \|_2 = \sqrt{ \left< f, f \right>} = \sqrt{\E_x[ f(x)^2 ]}$
and the $\ell_2$ distance between two functions $f, g : \ftwo^n \to \R$ is 
the $\ell_2$ norm of the function $f - g$.
In other words, $\|f - g \|_2 = \sqrt{\left< f-g, f-g \right>} 
=  \sqrt{\frac1{2^n}\sum_{x \in \ftwo^n} (f(x) - g(x))^2}$.

For $x,y\in \ftwo^n$ we denote the inner product as $x \cdot y = \sum_{i=1}^n x_i y_i$.
For $\alpha \in \ftwo^n$, the \emph{character} 
$\chi_\alpha : \ftwo^n \to \oo$ is the function defined by
$
\chi_\alpha(x) = (-1)^{\alpha \cdot x}.
$
Characters form an orthonormal basis as $\langle \chi_\alpha, \chi_\beta \rangle = \delta_{\alpha\beta}$ where $\delta$ is the Kronecker symbol.
The \emph{Fourier coefficient} of $f : \ftwo^n \to \R$ corresponding to $\alpha$ is
$
\hat{f}(\alpha) = \E_x[ f(x) \chi_\alpha(x)].
$
The \emph{Fourier transform} of $f$ is the function $\hat{f} : \ftwo^n \to \R$
that returns the value of each Fourier coefficient of $f$. 
The Fourier $\ell_1$ norm, or the \emph{spectral norm} of $f$, is defined as $\|\hat{f}\|_1 := \sum_{\alpha \in \ftwo^n} |\hat{f}(\alpha)|$.

\begin{fact}[Parseval's identity]
	\label{parseval}
	For any $f : \ftwo^n \to \R$ it holds that
	$
	\|f \|_2 = \| \hat{f} \|_2 
	= \sqrt{ \sum_{\alpha \in \ftwo^n} \hat{f}(\alpha)^2 }.
	$
	Moreover, if $f : \ftwo^n \to \oo$ then $\|f\|_2 = \|\hat f\|_2 = 1$.
\end{fact}

\subsubsection{Fourier $\ell_1$-Sampling}\label{sec:fourier-sampling}

The following Fourier $\ell_1$-sampling primitive is based on the work of Bruck and Smolensky~\cite{B92} (see also~\cite{G97,MO09}). 
Here we need to analyze its properties for approximating real-valued functions instead of computing Boolean functions as in~\cite{G97,MO09}.

\begin{proposition}[Fourier $\ell_1$-sampling]\label{prop:l1-sampling}
	For any $f \colon \ftwo^n \rightarrow \mathbb R$ it holds that $\rla{\eps}(f) = O(\|\hat f\|_1^2/\eps)$. 
\end{proposition}
\begin{proof}
	Sample $\mathbf{S} \in \{0,1\}^n$ from the following distribution: $\Pr[\mathbf{S} = S] = \frac{|\hat f(S)|}{\|\hat f\|_1}$.
	Let $Z = sgn(\hat f(\mathbf{S}))\chi_{\mathbf{S}}(x) \|\hat f\|_1$.
	Then: 
	\begin{align*}
	&&\E[Z] =& \E_{\mathbf{S}}[sgn(\hat f(\mathbf{S}))\chi_{\mathbf{S}}(x) \|\hat f\|_1] \\
	&&=& \sum_{S \in \{0,1\}^n} sgn(\hat f(S)) \frac{|\hat f(S)|}{\|\hat f\|_1} \chi_S(x) \|\hat f\|_1 \\
	& &=& \sum_{S \in \{0,1\}^n} \hat f(S) \chi_S(x) \\
	 &&=& f(x).
	\end{align*}
	Variance of $Z$ is:
	\begin{align*}
	&&Var[Z] =& \E_{\mathbf{S}}\left[\left(sgn(\hat f(\mathbf{S}))\chi_{\mathbf{S}}(x) \|\hat f\|_1 - f(x)\right)^2\right] \\ 
	&&=& \|\hat f\|_1^2 + f(x)^2 - 2 \|\hat f\|_1 f(x) \E_{\mathbf{S}}[sgn(\hat f(\mathbf{S})) \chi_{\mathbf{S}}(x)] \\
	&&=& \|\hat f\|_1^2 - f(x)^2 \\
	&&\le& \|\hat f\|_1^2.
	\end{align*}
	
	Thus averaging $Z$ over $\frac{\|\hat f\|_1^2}{\epsilon}$ repetitions reduces variance to at most $\epsilon$ as desired.
\end{proof}

It follows from Proposition~\ref{prop:l1-sampling} that additive and coverage functions admit small approximate $\ftwo$-sketches.
\begin{corollary}\label{cor:additive}
Let $\ell_w(x) \colon \{0,1\}^n \rightarrow \mathbb R$ be an additive function $\ell_w(x) = \sum_{i = 1}^n w_i x_i$. Then $$\rla{\eps}(\ell_w) = O(\min(\|w\|_1^2/\eps, n)).$$
\end{corollary}
\begin{proof}
Note that $\|\hat \ell_w\|_1 = O(\|w\|_1)$ and hence the bound follows.
\end{proof}

\begin{corollary}\label{cor:coverage}
If $f \colon \ftwo^n \to [0,1]$ is a coverage function then $\rla{\eps}(f) = O(1/\eps)$.
\end{corollary}
\begin{proof}
It is known (see Lemma 3.1 in~\cite{FK14}) that for such coverage functions $\|\hat f\|_1 \le 2$ and hence the desired bound follows from Proposition~\ref{prop:l1-sampling}.
\end{proof}

However, direct Fourier $\ell_1$-sampling can fail even in some fairly basic situations, e.g. even for budget-additive functions.
Consider, for example, the ``hockey stick'' function: $hs_{\frac12}(x) = \min\left(\frac{1}{2}, \frac1n \sum_{i = 1}^n x_i\right)$.
Fourier spectrum of this function is well-understood (see. e.g.~\cite{FV15}) and in particular $\|\hat f\|_1 = 2^{\Omega(n)}$.
Nevertheless small sketches for budget-additive functions can be constructed using the following composition theorem.

A function $f \colon \mathbb R^n \to \mathbb R$ is $\alpha$-Lipschitz if $|f(x) - f(y)| \le \alpha \|x-y\|_2$ for any $x,y \in \mathbb R^n$ and some constant $\alpha > 0$\footnote{Note that this definition is slightly different from the corresponding definition for functions over the Boolean hypercube}. 
\begin{proposition}[Composition theorem]\label{prop:composition}
	If $h \colon \mathbb R^t \to \mathbb R$ is an $\alpha$-Lipschitz function then for any functions $f_1, \dots, f_t$ where $f_i \colon \ftwo^n \to \mathbb R$ it holds that:
	$$\rla{\eps}(h(f_1,\dots, f_t)) \le \sum_{i = 1}^t \rla{\eps / \alpha^2 t}(f_i).$$
\end{proposition}
\begin{proof}
	Let $f'_1, \dots, f'_t$ be the sketches of $f_1, \dots, f_t$ respectively. Applying $h$ to their values we have:
	\begin{align*}
	\E[(h(f'_1, \dots, f'_t) - h(f_1, \dots, f_t))^2] \le \E[\alpha^2 \|f' - f\|_2^2] 
	= \alpha^2 \sum_{i = 1}^t \E[(f'_i(x) - f_i(x))^2] 
	\le \eps. \qedhere
	\end{align*}
\end{proof}
From Corollary~\ref{cor:additive} and Proposition~\ref{prop:composition} the following bound on approximate $\ftwo$-sketch complexity of budget-additive functions follows immediately. 
\begin{corollary}\label{cor:budget-additive}
	For any budget additive function $f(x) = \min(b, \sum_{i = 1}^n w_i x_i)$ it holds that:
	$$\rla{\eps}(f) = O(\min(\|w\|_1^2/\eps, n)).$$ 
\end{corollary}
\begin{proof}
In the composition theorem above, we set $h = \min(b, z)$ and hence $h$ is a $1$-Lipschitz function of $z$. Hence we can set $\alpha = 1$ and $t = 1$ by treating $\sum_{i = 1}^n w_i x_i$ as $f_1$ and the proof follows.
\end{proof}

\subsubsection{Optimality of Fourier $\ell_1$-Sampling}\label{sec:fourier-l1-optimal}

Let $\ell_w(x) \colon \{0,1\}^n \rightarrow \mathbb R$ be an additive function $\ell_w(x) = \sum_{i = 1}^n w_i x_i$ parametrized by $w \in \mathbb R^n$ with non-negative weights $w_1, w_2 \dots, w_n$. 
The corresponding XOR-function $\ell_w^+(x,y)$ gives weighted Hamming distance between vectors $x$ and $y$. The following result can be seen as a generalization of the unweighted Gap Hamming lower bound due to Jayram, Kumar and Sivakumar~\cite{JKS08} (see also~\cite{IW03,CR12}).

\begin{theorem}\label{thm:weighted-linear}
	For any additive function $\ell_w$ of the form $\ell_w(x) = \sum_{i = 1}^n w_i x_i$ where $w_i \ge 0$  it holds that for any $\epsilon \ge \|w\|_2^2$:
		$$\rca{\epsilon}(\ell_w^+) = \Omega\left(\frac{\|w\|_1^2}{\epsilon} \right).$$
\end{theorem}
\begin{proof}
	We use reduction from the standard communication problem $\Index$. 
	In this problem Alice is given $a \in \{0,1\}^n$ and Bob is given $t \in [n]$. Alice needs to send one message to Bob so that he can compute $a_t$. It is well-known that this requires linear communication:
	\begin{theorem}[\cite{KNR99}]
	$\rc{1/3}(\Index) = \Omega(n).$
    \end{theorem}
	
	Let $n$ be odd and $k$ be a parameter to be chosen later.
	Consider an instance of indexing where Alice has an input $a \in \{-1,1\}^k$ and Bob has an index $t \in [k]$.
	Draw $n$ random vectors $r_1, \dots, r_n$ where each $r_i$ is uniform over $\{-1,1\}^k$.
	Construct vectors $x,y \in \{-1,1\}^n$ as follows:
	\begin{align*}
	x_i = sign(\langle a, r_i\rangle), \quad\quad\quad y_i = sign(r_{i,t}), 
	\end{align*}
	where we define $sign(z) = -1$ if $z \le 0$ and $sign(z) = 1$ if $z > 0$.
	
	Note that if $a_t = 1$ then $\Pr[x_i = y_i] \ge \frac12 + \frac{c}{\sqrt{k}}$, otherwise $\Pr[x_i = y_i] \le \frac12 - \frac{c}{\sqrt{k}}$ for some absolute constant $c > 0$.
	Now consider the function $\ell_w^+(x,y) = \sum_{i = 1}^n w_i (x_i + y_i)$.
	We will show that for a suitable choice of $k$ with a large constant probability $\ell_w^+(x,y) > \frac12 \|w\|_1 + 2 \sqrt{\eps}$ if $a_t = 1$ and $\ell_w^+(x,y) < \frac12 \|w\|_1 - 2  \sqrt{\eps} = $ if $a_t = -1$.
    By Markov's inequality, a communication protocol for $\ell^+_w$ with expected squared error $\eps$ has squared error at most $4 \eps$ (and hence absolute error at most $2\sqrt{\eps}$) with probability at least $3/4$.
    Hence, such a protocol can distinguish these two cases with probability $3/4 - \xi$ where $\xi$ is the error probability introduced by the reduction.  If $\xi < 1/12$ then it can solve indexing on strings of length $k$ with probability at least $2/3$ and so a lower bound of $\Omega(k)$ follows.
	
	Indeed, consider the case $a_t = -1$, as the case $a_t = 1$ is symmetric.
	Let $Z_i$ be a random variable defined as $Z_i = w_i I[x_i = y_i]$.
	We have $\E[Z_i] \le w_i \left(\frac12 - \frac{c}{\sqrt{k}}\right)$.
	Let $Z = \sum_{i = 1}^n Z_i$, then:

	\begin{align*}
	\E[Z] \le \sum_{i = 1}^n w_i \left(\frac12 - \frac{c}{\sqrt{k}}\right) = \|w\|_1\left(\frac{1}{2} - \frac{c}{\sqrt{k}}\right). 
	\end{align*}
	
	Let $X_i = Z^{\le i} - \E[Z^{\le i}]$ where $Z^{\le i} = \sum_{j = 1}^i Z_j$.
	We have 
	\begin{align*}
	\E[X_{i + 1} | X_1, \dots,  X_i] &= \E[Z^{\le i + 1} - \E[Z^{\le i + 1}] | X_1, \dots, X_i] \\
	&= \E[Z_{i + 1} - \E[Z_{i + 1}] + X_i | X_1, \dots, X_i] \\
	&= \E[Z_{i + 1} - \E[Z_{i + 1}]] + X_i \\
	& =X_i,
	\end{align*}
	and hence $X_i$ is a martingale. Furthermore, for every $i$ it holds that:
	 $$|X_i - X_{i - 1}| =  |Z^{\le i} - \E[Z^{\le i}] - Z^{\le i - 1} + \E[Z^{\le i - 1}]| = |Z_i - \E[Z_i]| < |w_i|.$$
	We can now use the following form of Azuma's inequality:
	\begin{theorem}[Azuma's inequality]
	If $X_i$ for $i = 0, 1, \dots $ is a martingale such that $X_0 = 0$ and $|X_i - X_{i - 1}| < c_i$ almost surely then for every integer $m$ and positive real $\theta$ it holds that:
	$$\Pr[X_m  \ge \theta] \le e^{- \frac{\theta^2}{2 \sum_{i = 1}^m c_i^2}}.$$
	\end{theorem}

	Applying Azuma's inequality we have: $\Pr[X_n \ge \theta] \le e^{ - \frac{\theta^2}{2 \|w\|_2^2}}.$ Recall that $\E[Z^{\le n}] \le \|w\|_1\left(\frac{1}{2} - \frac{c}{\sqrt{k}}\right)$ and hence:
	$$\Pr\left[Z \ge \|w\|_1\left(\frac{1}{2} - \frac{c}{\sqrt{k}}\right) + \theta\right] \le e^{- \theta^2/2 \|w\|_2^2}.$$
    Setting $\theta = \frac{c\|w\|_1}{2\sqrt{k}}$ we have $\Pr\left[Z \ge \frac{\|w\|_1}{2} (1 - c/\sqrt{k})\right] \le e^{-\frac{c^2 \|w\|_1^2}{8k\|w\|_2^2}}$. If $k = \frac{c^2 \|w\|_1^2}{36\|w\|_2^2}$ then: 
    $$\Pr\left[Z \ge \frac{\|w\|_1}{2} -  3\|w\|_2\right] \le e^{-4}.$$

Using similar analysis for the case $a_t = 1$ we conclude that with probability at least $1 - 2 e^{-4} > 1 - 1/12$ in this case  it holds that $\Pr\left[Z \le \frac{\|w\|_1}{2} + 3\|w\|_2\right] \le e^{-4}$ and hence error probability $\xi$ introduced by the reduction is at most $1/12$. 
Thus using this reduction we obtain a protocol for solving indexing on strings of length $k$ with probability at least $2/3$ and the lower bound of $\Omega(k) = \Omega(\|w\|_1^2/\|w\|^2_2) = \Omega(\|w\|_1^2/ \epsilon)$ follows where we used the fact that $\epsilon \ge \|w\|_2^2$.
\end{proof}

\subsection{Information Theory} \label{app:information-theory}
Let $X$ be a random variable supported on a finite set $\{x_1, \ldots, x_s\}$. Let $\mathcal{E}$ be any event in the same probability space. Let $\mathbb{P}[\cdot]$ denote the probability of any event. The \emph{conditional entropy} $H(X \mid \mathcal{E})$ of $X$ conditioned on $\mathcal{E}$ is defined as follows.
\begin{definition}[Conditional entropy]
	\[H(X \mid \mathcal{E}):=\sum_{i=1}^s \mathbb{P}[X=x_i \mid \mathcal{E}]\log_2 \frac{1}{\mathbb{P}[X=x_i \mid \mathcal{E}]}\]
\end{definition}
An important special case is when $\mathcal{E}$ is the entire sample space. In that case the above conditional entropy is referred to as the \emph{Shannon entropy} $H(X)$ of $X$.
\begin{definition}[Entropy]
	\[H(X):=\sum_{i=1}^s\mathbb{P}[X=x_i] \log_2 \frac{1}{\mathbb{P}[X=x_i]}\]
\end{definition}
Let $Y$ be another random variable in the same probability space as $X$, taking values from a finite set $\{y_1, \ldots, y_t\}$. Then the conditional entropy of $X$ conditioned on $Y$, $H(X \mid Y)$, is defined as follows.
\begin{definition}
	\[H(X \mid Y)=\sum_{i=1}^t \mathbb{P}[Y=y_i] \cdot H(X \mid Y=y_i)\]
\end{definition}
We next define the binary entropy function $H_b(\cdot)$.
\begin{definition}[Binary entropy]
	For $p \in (0,1)$, the binary entropy of $p$, $H_b(p)$, is defined to be the Shannon entropy of a random variable taking two distinct values with probabilities $p$ and $1-p$.
	\[H_b(p):=p \log_2 \frac{1}{p} + (1-p) \log \frac{1}{1-p}.\]
\end{definition}
The following properties of entropy and conditional entropy will be useful.
\begin{fact}
	\label{fct0}
	\begin{enumerate}
		\item[\emph{(1)}] \label{support} Let $X$ be a random variable supported on a finite set $\mathcal{A}$, and let $Y$ be another random variable in the same probability space. Then $0 \leq H(X \mid Y) \leq H(X)  \leq \log_2 |\mathcal{A}|$.
		\item[\emph{(2)}] \label{subadditivity}\emph{(Sub-additivity of conditional entropy)}.   Let $X_1, \ldots, X_n$ be $n$ jointly distributed random variables in some probability space, and let $Y$ be another random variable in the same probability space, all taking values in finite domains. Then,
		\[H(X_1, \ldots, X_n \mid Y) \leq \sum_{i=1}^n H(X_i \mid  Y).\]
		\item[\emph{(3)}] Let $X_1, \ldots, X_n$ are independent random variables taking values in finite domains. Then,
		\[H(X_1, \ldots, X_n) = \sum_{i=1}^n H(X_i).\]
		\item[\emph{(4)}]\emph{(Taylor expansion of binary entropy in the neighborhood of $\frac{1}{2}$).} \[H_b(p)=1-\frac{1}{2 \log_e 2}\sum_{n=1}^\infty \frac{(1-2p)^{2n}}{n(2n-1)} \]
	\end{enumerate}
\end{fact}
\begin{definition}[Mutual information]
	Let $X$ and $Y$ be two random variables in the same probability space, taking values from finite sets. The mutual information between $X$ and $Y$, $I(X;Y)$, is defined as follows.
	\[I(X;Y):=H(X)-H(X \mid Y).\]
	It can be shown that $I(X;Y)$ is symmetric in $X$ and $Y$, i.e.  $I(X;Y)=I(Y;X)=H(Y)-H(Y \mid X)$.
\end{definition}
The following observation follows immediately from the first inequality of Fact~\ref{fct0} ($1$).
\begin{fact}
	\label{I>=0}
	For any two random variables $X$ and $Y$, $I(X;Y) \leq H(X)$.
\end{fact}

\section{Missing Proofs}
\label{app:missing}

\subsection{Proof of Theorem~\ref{thm:approx-f2-sketch-uniform}}
\label{app:thm:approx-f2-sketch-uniform}
\begin{proofof}{Theorem~\ref{thm:approx-f2-sketch-uniform}}
\begin{proof}
	The proof is largely based on a similar proof in~\cite{KMSY18} except that here we work with real-valued functions with unbounded norm.
	In the next two lemmas, we look into the structure of a one-way communication protocol for $f^+$, and analyze its performance when the inputs are uniformly distributed. We give a lower bound on the number of bits of information that any correct randomized one-way protocol reveals about Alice's input\footnote{We thus prove an \emph{information complexity} lower bound. See, for example, \cite{DBLP:conf/pods/Jayram10} for an introduction to information complexity.}, in terms of the linear sketching complexity of $f$ for uniform distribution.
	
	The next lemma bounds the probability of error of a one-way protocol from below in terms of the Fourier coefficients of $f$, and the conditional distributions of different parities of Alice's input conditioned on Alice's random message.
	\begin{lemma}
		\label{lem:fourier-error}
		Let $\epsilon \in [0, \frac{1}{2})$. Let $\Pi$ be a deterministic one-way protocol for $f^+$ such that  
		$\E_{x,y \sim U(\ftwo^n)}[\Pi(x,y) - f^+(x,y)]^2 \leq \epsilon$. 
		Let $M$ denote the distribution of the random message sent by Alice to Bob in $\Pi$.   For any fixed message $m$ sent by Alice, let $\D_m$ denote the distribution of Alice's input $x$ conditioned on the event that $M=m$. Then,
		
		\[\epsilon \geq \sum_{\alpha \in \ftwo^n} \widehat{f}(\alpha)^2\cdot \left(1-\E_{m \sim M}\left(  \E_{x \sim \D_m} [\rchi_\alpha(x)]\right)^2\right).\]
	\end{lemma}
	\begin{proof}
		For any fixed input $y$ of Bob, define
		$\epsilon_m^{(y)}:=\E_{x \sim \D_m}(\Pi(x,y) - f^+(x,y))^2$. Thus,
		\begin{align}
		\label{error1}
		\epsilon \geq \E_{m \sim M} \E_{y \sim U(\ftwo^n)} [\epsilon_m^{(y)}].
		\end{align}
		Note that the output of the protocol is determined by Alice's message and $y$. Hence for a fixed message and Bob's input, if the restricted function has high variance, the protocol is forced to commit error with high probability. Formally, let $a_m^{(y)}$ be the output of the protocol when Alice's message is $m$ and Bob's input is $y$. Also, define $\mu_m^{(y)} :=\E_{x \sim \D_m} [f^+ (x,y)]$. Then, 
		\begin{align} 
		\label{error2}\epsilon_m^{(y)} &=\E_{x \sim \D_m} \left[(a_m^{(y)} - f^+(x,y))^2\right] \nonumber \\
		&= \E_{x \sim \D_m}\left[((\mu_m^{(y)} - f^+(x,y))+(a_m^{(y)}- \mu_m^{(y)}))^2\right] \nonumber \\
		&= \E_{x \sim \D_m}\left[((\mu_m^{(y)} - f^+(x,y))^2+(a_m^{(y)} - \mu_m^{(y)})^2)\right] + 2 (a_m^{(y)}-\mu_m^{(y)}) \E_{x \sim \D_m} \left[(\mu_m^{(y)} - f^+(x,y))\right] \nonumber \\
		&\geq  \E_{x \sim \D_m}\left[(\mu_m^{(y)} - f^+(x,y))^2\right] \nonumber \\
		&=\mathsf{Var}_{x \sim \D_m} \left[f^+(x,y)\right].
		\end{align}
		
		Now,
		\begin{align*}
		& \mathrm{Var}_{x \sim \D_m}[f^+(x,y)] = \E_{x \sim \D_m} [f^+(x,y)^2] - \left( \E_{x \sim \D_m} [f^+(x,y)]\right)^2 \\
		&=\E_{x \sim \D_m} [f^+(x,y)^2]-\left(\sum_{\alpha \in \ftwo^n} \widehat{f}(\alpha)\rchi_\alpha(y) \E_{x \sim \D_m} [\rchi_\alpha(x)]\right)^2
		\\
		&=\E_{x \sim \D_m} [f^+(x,y)^2]-\left(  \sum_{\alpha \in \ftwo^n}  \widehat{f} (\alpha)^2 \left(  \E_{x \sim \D_m} [\rchi_\alpha(x)]\right)^2 \right. \\
		&\qquad \qquad \left. + \sum_{(\alpha_1, \alpha_2) \in \ftwo^n \times \ftwo^n: \alpha_1 \neq \alpha_2} \widehat{f}(\alpha_1) \widehat{f}(\alpha_2) \rchi_{\alpha_1 + \alpha_2}(y)\E_{x \sim \D_m}[\rchi_{\alpha_1}(x)]\E_{x \sim \D_m}[\rchi_{\alpha_2}(x)]\right).
		\end{align*}
		Taking expectation over $y$ we have:
		\begin{align*}
		\E_{y \sim U(\ftwo^n)} \left[\mathrm{Var}_{x \sim \D_m}[f^+(x,y)] \right]
		&=\E_{y \sim U(\ftwo^n)}\E_{x \sim \D_m} [f^+(x,y)^2]-\sum_{\alpha \in \ftwo^n} \widehat{f}(\alpha)^2 \left(  \E_{x \sim \D_m} [\rchi_\alpha(x)]\right)^2 \\
		&=\E_{x \sim \D_m} \E_{y \sim U(\ftwo^n)} [f^+(x,y)^2]-\sum_{\alpha \in \ftwo^n} \widehat{f}(\alpha)^2 \left(  \E_{x \sim \D_m} [\rchi_\alpha(x)]\right)^2 \\
		&= \|f\|_2^2 -\sum_{\alpha \in \ftwo^n} \widehat{f}(\alpha)^2 \left(  \E_{x \sim \D_m} [\rchi_\alpha(x)]\right)^2,
		\end{align*}
		where in the last step we used the fact that for any fixed $x$ we have $\E_{y \sim U(\ftwo^n)} [f^+(x,y)^2] = \E_{z \sim U(\ftwo^n)}[f^2(z)] = \|f\|_2^2$.
		Taking expectation over messages it follows using (\ref{error1}), (\ref{error2}) that,
		\begin{align}
		\epsilon & \ge \|f\|_2^2 - \sum_{\alpha \in \ftwo^n} \widehat{f}(\alpha)^2 \cdot \E_{m \sim M} \left(  \E_{x \sim \D_m} [\rchi_\alpha(x)]\right)^2  \nonumber \\
		&=\sum_{\alpha \in \ftwo^n} \widehat{f}(\alpha)^2\cdot\left(1- \E_{m \sim M}\left(  \E_{x \sim \D_m} [\rchi_\alpha(x)]\right)^2\right). \nonumber \\
		\end{align}
		The second equality above follows from Parseval's identity. The lemma follows.
	\end{proof}
	Let $\Pi$ be a deterministic protocol such that $\mathbb E_{x,y \sim U(\ftwo^n)}[(\Pi(x,y) - f^+(x,y))^2] \leq \epsilon$, with optimal cost $c_\Pi:=\mathcal{D}^{\rightarrow,U}_{\epsilon}(f^+)$. To prove our theorem, we use the protocol $\Pi$ to come up with a subspace of $\ftwo^n$. Next, in Lemma~\ref{prop:cl1} (a) we prove, using  Lemma~\ref{lem:fourier-error}, that $f$ is $\xi$-concentrated on that subspace where $\xi = \|f\|_2^2 - \eps(1 + 2 \delta)$.  In Lemma~\ref{prop:cl1} (b) we upper bound the dimension of that subspace in terms of $c_\Pi$.
	
	Let $\mathcal{A}_\delta:=\{\alpha \in \ftwo^n:  \E_{m \sim M}\left(  \E_{x \sim \D_m}\rchi_\alpha(x)\right)^2 \geq \delta \} \subseteq \ftwo^n$.
	\begin{lemma}\label{prop:cl1}
		Let $\delta \in [0,1/2]$ and $\xi = \|f\|_2^2 - \epsilon (1 + 2 \delta)$, then
		$\sum_{\alpha \notin \mathsf{span}(\mathcal{A}_\delta)}\widehat{f}(\alpha)^2 \leq \|f\|_2^2 - \xi.$
	\end{lemma}
	\begin{proof}
		We show that $f$ is $\xi$-concentrated on $\mathsf{span(\mathcal{A}_\delta)}$. By Lemma~\ref{lem:fourier-error} we have that
		\begin{align*}
		\epsilon &\geq \sum_{\alpha \in \mathsf{span}(\mathcal{A}_\delta)}  \widehat{f}(\alpha)^2 \cdot\left(1- \E_{m \sim M}\left(  \E_{x \sim \D_m} \rchi_\alpha(x)\right)^2\right) \nonumber + \sum_{\alpha \notin \mathsf{span}(\mathcal{A}_\delta)} \widehat{f}(\alpha)^2 \cdot\left(1- \E_{m \sim M}\left(  \E_{x \sim \D_m} \rchi_\alpha(x)\right)^2\right)  \\
		& > (1-\delta) \cdot \sum_{\alpha \notin \mathsf{span}(\mathcal{A}_\delta) } \widehat{f}(\alpha)^2 .
		\end{align*}
		Thus $\sum_{\alpha \notin \mathsf{span}(\mathcal{A}_\delta)}\widehat{f}(\alpha)^2 < \frac{\epsilon}{1-\delta} \leq \epsilon \cdot (1+2\delta) = \|f\|_2^2 - \xi$ (since $\delta \leq 1/2$).
	\end{proof}
	
	Now we are ready to complete the proof of Theorem~\ref{thm:approx-f2-sketch-uniform}.	 
	Let $\ell = dim(\mathsf{span}(\mathcal{A}_\delta))$. Then it suffices to show that $\ell\leq \frac{2c_\Pi}{\delta}.$	
	Note that $\rchi_\alpha(x)$ is a unbiased random variable taking values in $\{1, -1\}$. For each $\alpha$ in the set $\mathcal{A}_\delta$ in Proposition~\ref{prop:cl1}, the value of  $\E_{m \sim M}\left(\E_{x \sim \D_m} \rchi_\alpha(x)\right)^2$ is bounded away from $0$. This suggests that for a typical message $m$ drawn from $M$, the distribution of $\rchi_\alpha(x)$ conditioned on the event $M=m$ is significantly biased. Fact~\ref{entropy} enables us to conclude that Alice's message reveals $\Omega(1)$ bit of information about $\rchi_\alpha(x)$. However, since the total information content of Alice's message is at most $c_\Pi$, there can be at most $O(c_\Pi)$ independent vectors in $\mathcal{A}_\delta$. Now we formalize this intuition.
	
	In the derivation below we use several standard facts about properties of entropy and mutual information which can be found in Appendix~\ref{app:information-theory}.
	We will need the following fact about entropy of a binary random variable. The proof can be found in Appendix A of~\cite{KMSY18}.
	\begin{fact}
		\label{entropy}
		For any random variable $X$ supported on $\{1, -1\}$, $H(X) \leq 1-\frac12(\mathbb{E} X)^2$.
	\end{fact}

	Let $\mathcal{T}=\{\alpha_1, \ldots, \alpha_\ell\}$ be a basis of $\mathsf{span}(\mathcal{A}_\delta)$. Then,
	\begin{align*}
	c_\Pi &\geq H(M)\qquad \qquad \qquad \qquad \qquad \qquad\\
	& \geq I(M; \rchi_{\alpha_1}(x), \ldots, \rchi_{\alpha_\ell}(x)) \qquad \qquad\\
	&=H(\rchi_{\alpha_1}(x), \ldots, \rchi_{\alpha_{\ell}}(x))-H(\rchi_{\alpha_1}(x), \ldots, \rchi_{\alpha_{\ell}}(x) \mid M)  \\
	& = \ell -  H(\rchi_{\alpha_1}(x), \ldots, \rchi_{\alpha_{\ell}}(x) \mid M)\qquad\\
	& \ge \ell- \sum_{i=1}^\ell H(\rchi_{\alpha_i}(x) \mid M) 
	\qquad \qquad \qquad \\
	& \ge\ell - \sum_{i = 1}^\ell (1 - \frac{1}{2} \left(\E[\rchi_{\alpha_i}(x) | M])^2\right) \qquad \qquad \qquad\mbox{\ \ \ \ \ \ \ \  (by Fact~\ref{entropy})}\\
	& \geq \ell -\ell \left(1-\delta\cdot\frac{1}{2}\right) \nonumber \\
	&=\frac{\ell\delta}{2}.
	\end{align*}
	Thus $\ell \leq \frac{2c_\Pi}{\delta}$.
\end{proof}
\end{proofof}

\section{Subadditive Functions}\label{app:xs-functions}
\begin{definition}
A function $f:2^{[n]}\to\mathbb{R}^+$ is \emph{subadditive} if $f(A\cup B)\le f(A)+f(B)$ for all $A,B\subseteq[n]$.
\end{definition}
The class of XS functions introduced in~\cite{LLN06} are subadditive functions that correspond to unit demand functions $f(S) = \max_{i \in S} w_i$.  
Similarly, a subadditive function $f$ is XOS if $f$ can be expressed as the maximum of a number of linear combinations of valuations, where each weight in the linear combination is positive, $f(S)=\max_{1\le i\le k}w_i^\top\chi(S)$, where $w_{i,j}\ge 0$ for all $i\in[n]$. 
By flipping the roles of the MAX and SUM operators, we obtain a strict subclass of XOS valuations, called OXS functions. 

It is known that OXS functions is a strict subset of submodular functions, which is a strict subset of XOS functions, which is a strict subset of subadditive functions~\cite{LLN06}. 
For more details on subadditive functions, see \cite{BCIW12}. 

\subsection{Lower Bound for XS Functions}
\begin{theorem}
If $f$ is an XS function corresponding to a collection of distinct weights then $\rc{1/3}(f^+) = \Omega(n)$.
\end{theorem}
\begin{proof}
Let $w_1 > w_2 > \dots > w_n$.
We use a reduction from a standard communication problem called Augmented Indexing, denoted $\mathsf{AI}(x,i)$.
In this problem Alice's input is $x \in \ftwo^n$ and Bob's input is $i \in[n]$ and the bits $x_1, \dots, x_{i - 1}$.

\begin{theorem}[\cite{MNSW98,BYJKK04}] 
	$\rc{1/3}(\mathsf{AI}) = \Omega(n)$.
\end{theorem}

In order to solve $\mathsf{AI}(x,i)$ using a protocol for $f^+$ set $x' = x$ and $y' = (x_1, \dots, x_{i - 1}, 0, \dots, 0)$. 
If $\mathsf{AI}(x,i) = 1$ then $f^+(x' + y') = w_i$, otherwise $f^+(x' + y') \le w_{i + 1}$. Hence an $\Omega(n)$ lower bound follows.

\end{proof}

\input{uniform}

\section{Sketches Under Uniform Distribution}\label{app:uniform-sketches}
Recall that a function $f:\{0,1\}^n\to\{0,1\}$ is $\eps$-approximated by a function $g:\{0,1\}^n\to\{0,1\}$ if
\[\mathbf{Pr}_{x\in\{0,1\}^n}[f(x)\neq g(x)]\le\eps.\]
\begin{theorem}
\label{thm:junta:approx}
\cite{LovettZ18}
Every DNF with width $w$ can be $\eps$-approximated by a $\left(\frac{1}{\eps}\right)^{O(w)}$-junta.
\end{theorem}

\begin{theorem}
Let $f$ be a $(\theta,m)$-LTF of width $w$ DNFs then $\mathcal{D}^{lin,U}_{1-\eps}(f)\le u\left(\frac{1}{\eps},\frac{\theta}{m},w\right)$ for some function $u$.
\end{theorem}
\begin{proof}
Observe that by Lemma~\ref{lem:disjunction} and standard probability amplification techniques, there exists an $\frac{\eps}{2}$-approximation of $f$, denoted $f_h$, that is a threshold function of $q=O\left(\left(\frac{\theta}{m}\right)^3\log^2\frac{\theta}{m}\log\frac{2}{\eps}\right)$ width $w$ disjunctions. 
Thus, it suffices to take a $\frac{\eps}{2q}$-approximation for each of the $q$ width $w$ disjunctions, using Theorem~\ref{thm:junta:approx}. 
By a simple union bound, the resulting linear sketch differs from $f$ on at most $\eps$ fraction of the inputs.
The resulting sketch complexity is $\left(\log\frac{1}{\eps}\right)^{O(w)}$ for each of the $O\left(\left(\frac{\theta}{m}\right)^3\log^2\frac{\theta}{m}\right)$ disjunctions.
\end{proof}

\begin{corollary}
If $f$ can be represented as a (monotone) linear threshold function of $(\theta_i,m_i)$-linear threshold functions, then $\mathcal{D}_{1-\eps}^{lin,U}(f)\le u\left(\frac{1}{\eps},\frac{\theta}{m}\right)$, where $\frac{\theta}{m}=\max_i\frac{\theta_i}{m_i}$.
\end{corollary}
Note that any matroid of rank $r$ can be expressed as a linear threshold function of DNFs, where each clause contains the $r$ independent basis elements. 
Therefore, matroid rank functions can be sketched succinctly under the uniform distribution:
\begin{theorem}
Given $0<\eps<1$ to be the probability of failure, matroids of rank $r$ have a randomized linear sketch of size $g\left(r,\frac{1}{\eps}\right)$ under the uniform distribution, where $g$ is some function. 
\end{theorem}
In fact, the function $f(\cdot)$ can be improved using the following observation about using juntas to approximate monotone submodular functions.
\begin{theorem}
[\cite{BOSY13}, Theorem 6]
Let $f:\{0,1\}^n\to\{a_1,\ldots,a_r\}$ be a monotone submodular function. 
For any $\eps\in\left(0,\frac{1}{2}\right)$, $f$ is $\eps$-close to a $(z+1)^{r+1}$-junta, where $z=O\left(r\log r+\log\frac{1}{\eps}\right)$.
\end{theorem}
\begin{corollary}\label{cor:matroid-approximation-uniform-rank}
\label{thm:uniform:matroid:exact}
Matroids of rank $r$ under the uniform distribution have a sketch of size $O\left(\left(r\log r+\log\frac{1}{\eps}\right)^{r+1}\right)$, where $\eps$ is the probability of failure.
\end{corollary}
\noindent
Finally, we note the concentration of matroid rank functions on uniform distributions (see also~\cite{V10}):
\begin{theorem}
[\cite{V10},\cite{GHRU13} Lemma 2.1]
Let $f:2^U\to\mathbb{R}$ be a $1$-Lipschitz submodular function. 
Then for any product distribution $\mathcal{D}$ over $2^U$,
\[\Pr_{S\sim\mathcal{D}}\left[|f(x)-\mathbb{E}[f(S)]|\ge t\right]\le2\exp\left(-\frac{t^2}{2\left(\mathbb{E}[f(S)]+5t/6\right)}\right).\]
\end{theorem}
\begin{corollary}\label{cor:matroid-approximation-uniform}
For matroids of rank $\Omega(n)$ and $\epsilon>\frac{1}{\sqrt{n}}$, the expectation of the matroid rank function $rank_M$ suffices for a $\eps$-approximation to the matroid rank. 
\end{corollary}
\begin{proof}
Recall that the matroid rank function $rank_M$ is always a submodular $1$-Lipschitz function. 
Thus, matroids of rank $\Omega(n)$ are highly concentrated around their expectation and so to get an $\eps$-approximation to the matroid rank, it suffices to simply output the expectation of $f$, provided $\epsilon>\frac{1}{\sqrt{n}}$. 
\end{proof}

%% file: uniform.tex
\section{Communication Under the Uniform Distribution}
\label{sec:uniform:hs:cc}
In this section, we switch to lower bounds for the uniform distribution and show the following result for the ``hockey stick'' function:
\begin{theorem}\label{thm:lb-hockey-stick-uniform}
	For any odd $n$, constant $c > 0$ and $\alpha = c \sqrt{n}$ there exists a constant $\epsilon > 0$ such that for the``hockey stick'' function $hs_{\alpha}(x) = \min(\alpha, \frac{2\alpha}n \sum_{i = 1}^n x_i)$ it holds that:
	$$\distcu{\epsilon}(hs_{\alpha}^+) = \Omega(n)$$
\end{theorem}

The proof relies on the characterization of communication complexity using approximate Fourier dimension (Theorem~\ref{thm:approx-f2-sketch-uniform}).
We also require a structural result, whose proof we defer to Section~\ref{app:hs}, about the Fourier spectrum of the hockey stick function.
\begin{lemma}\label{lem:hs-spectrum}
Let $n$ be odd and let  $hs_\alpha(x) = min(\alpha, \frac{2\alpha} n \sum_{i = 1}^n x_i)$ then:
$$\|\widehat{hs_\alpha}\|_2^2 - \widehat{hs_\alpha} (\emptyset)^2 - \widehat{hs_\alpha} ([n])^2 =  \Theta\left(\frac{\alpha^2}{n}\right),$$
where $[n]$ denotes the set $\{1, \dots, n\}$.
\end{lemma}

\noindent
We are now ready to prove Theorem~\ref{thm:lb-hockey-stick-uniform}.
\begin{proofof}{Theorem~\ref{thm:lb-hockey-stick-uniform}}
	By Lemma~\ref{lem:hs-spectrum} it follows that $\sum_{S \neq \emptyset, S \neq [n]} \widehat{hs_\alpha}(S)^2 = \Omega\left(\frac{\alpha^2}{n}\right) = \Omega(1)$, as $\alpha=c\sqrt{n}$.
	Since $hs_\alpha$ is a symmetric function and hence its Fourier coefficients for all sets of the same size are the same, one can show that it is not $\|hs_\alpha\|_2^2-\eps$-concentrated on $o(n)$-dimensional subspaces.
	Formally, this is proved in Theorem 4.6 in~\cite{KMSY18} which shows that there exists $\epsilon > 0$ such that $\dim_{\|f\|_2^2 - \eps}(f) = \Omega(n)$ for any symmetric function which satisfies the condition $\sum_{S \neq \emptyset, S \neq [n]} \hat f(S)^2 = \Omega(1)$.
\end{proofof}

\subsection{Fourier Spectrum of the ``Hockey Stick'' Function}\label{app:hs}
In this section, we bound the fourier spectrum of the ``hockey stick'' function.
\begin{proofof}{Lemma~\ref{lem:hs-spectrum}}
We have:
\begin{align*}
\|hs_\alpha\|_2^2 = 2^{-n} \sum_{x \in \{0,1\}^n} hs_\alpha (x)^2 &= 2^{-n} \left(\alpha^2 2^{n - 1}  + \frac{4\alpha^2}{n^2}\sum_{i = 0}^{\lfloor n/2 \rfloor} i^2 \binom{n}{i} \right)
\end{align*}
We also have:
\begin{align*}
\widehat{hs_\alpha} (\emptyset)^2 = \left(2^{-n} \sum_{x \in \{0,1\}^n} hs_\alpha(x) \right)^2 = 2^{-2n} \left(\alpha 2^{n-1} + \frac{2\alpha}{n}\sum_{i = 0}^{\lfloor n/2 \rfloor} i\binom{n}{i} \right)^2.
\end{align*}
Hence:
\begin{align*}
\|\widehat{hs_\alpha}\|_2^2 - \widehat{hs_\alpha} (\emptyset)^2 = 4 \alpha^2 \left(\frac{1}{16} + \frac{1}{n^2 2^n}  \sum_{i = 0}^{\lfloor n/2 \rfloor} i^2 \binom{n}{i} - \frac{1}{n 2^{n + 1}} \sum_{i = 0}^{\lfloor n/2 \rfloor } i\binom{n}{i} - \frac{1}{n^2 2^{2n}} \left(\sum_{i = 0}^{\lfloor n/2\rfloor} i\binom{n}{i}\right)^2\right)
\end{align*}

\noindent
For $i \ge 1$ we have $i \binom{n}{i} = i \frac{n!}{i! (n - i)!} = n \frac{(n -1)!}{(i - 1)!(n-i)!} = n \binom{n - 1}{i - 1}$.
Hence: 
$$\sum_{i = 0}^{\lfloor n/2 \rfloor} i \binom{n}{i} = n\sum_{i = 0}^{\lfloor n/2 \rfloor - 1} \binom{n - 1}{i} = n (2^{n - 2} - \binom{n - 1}{(n - 1)/2} / 2) \approx n2^{n - 2} (1 - \frac{\sqrt{2} }{\sqrt{\pi n}}).$$
where the approximation omits lower order terms. 
Similarly we have $i^2 \binom{n}{i} = ni \binom{n - 1}{i - 1} = n \binom{n - 1}{i - 1} + n (i - 1) \binom{n - 1}{i -1}$.
Hence: 

\begin{align*}\sum_{i = 0}^{\lfloor n/2 \rfloor} i^2\binom{n}{i}&=  n \sum_{i = 0}^{\lfloor n/2 \rfloor - 1}\binom{n - 1}{i} + n \sum_{i = 0}^{\lfloor n/2\rfloor - 1} i \binom{n - 1}{i} \\
&= n \left(2^{n - 1} - \binom{n - 1}{(n - 1)/2} / 2\right)  + n (n - 1) \sum_{i = 0}^{\lfloor n/2\rfloor - 2} \binom{n - 2}{i} \\
&= n \left(2^{n - 1} - \binom{n - 1}{(n - 1)/2} / 2\right) + n (n - 1)\left( 2^{n - 3} - \binom{n - 2}{\lfloor n - 2\rfloor} \right)\\
& \approx n \left(2^{n - 1} - \frac{\sqrt{2}2^{n - 2}}{\sqrt{\pi n}}\right) + n (n - 1) \left(2^{n - 3} - \frac{\sqrt{2} 2^{n - 2}}{\sqrt{\pi n}}\right)\\
& = n^2 2^{n - 3} - \frac{\sqrt{2} n^{3/2}2^{n - 2}}{{\sqrt{\pi}}} + 3n 2^{n - 3} + o(2^n n),
\end{align*}
where the approximation again omits lower order terms.

\noindent
Thus we have:
\begin{align*}
\|\widehat{hs_\alpha}\|_2^2 - \widehat{hs_\alpha} (\emptyset)^2 = \Theta\left(\frac{\alpha^2 }{n}\right)
\end{align*}

To complete the proof, we will show that $\widehat{hs_{\alpha}}([n])^2 = 0$. It is well-known (see e.g.~\cite{FV15}) that for all $S \subseteq [n]$ such that $|S| \ge 2$ and $i\in S$, it holds that $\widehat{hs_\alpha}(S) = \alpha \frac{\widehat{Maj}(S \setminus \{i\})}{n}$.
Using the fact that majority is an odd function, its Fourier coefficients on sets of even size are $0$.
\end{proofof}